\RequirePackage{etoolbox}
\csdef{input@path}%
{%
 {sty/},% class file
 {img/},% graphics files
}
\documentclass{elsevierbook}
\usepackage{natbib}
\usepackage{amssymb}
\usepackage{amsfonts}
\usepackage{mathrsfs}
\usepackage{graphicx}
\usepackage{hyperref}
\usepackage{cancel}
\usepackage{graphicx}
\usepackage{array}
\newcolumntype{P}[1]{>{\centering\arraybackslash}p{#1}}

\usepackage{xhfill}% http://ctan.org/pkg/xhfill
\usepackage{ulem}

\newcommand{\quotes}[1]{``#1''}

\usepackage{units}
\newcommand{\X}{\mathbb{X}}

\makeatletter
\newcommand{\vast}{\bBigg@{4}}
\newcommand{\Vast}{\bBigg@{5}}
\makeatother
\newtheorem{thm}{Theorem}
\newtheorem{theorem}[thm]{Theorem}
\newtheorem{lemma}[thm]{Lemma}
\newtheorem{definition}[thm]{Definition}

\newtheorem{corollary}[thm]{Corollary}
\begin{document}

%\Frontmatter
  %\include{titlepage}
  %\include{dedication}
  %\include{tableofcontents}
  \include{faketableofcontents}
  %\include{listoffigures}
  %\include{listoftables}
  %\include{preface}
  %\include{aboutauthors}

%\Mainmatter
  %\include{part1}
  \chapter{Information Geometry and Classical Cram\'{e}r-Rao Type Inequalities}\label{chap1}
\subchapter{Kumar Vijay Mishra$^\dag$ and M. Ashok Kumar$^\ddag$\\ \small{$^\dag$United States CCDC Army Research Laboratory, Adelphi, MD 20783 USA\\$^\ddag$Department of Mathematics, Indian Institute of Technology Palakkad, 678557 India}}

\minitoc
\begin{abstract}
We examine the role of information geometry in the context of classical Cram\'er-Rao (CR) type inequalities. In particular, we focus on Eguchi's theory of obtaining dualistic geometric structures from a divergence function and then applying Amari-Nagoaka's theory to obtain a CR type inequality. The classical deterministic CR inequality is derived from Kullback-Leibler (KL)-divergence. We show that this framework could be generalized to other CR type inequalities through four examples: $\alpha$-version of CR inequality, generalized CR inequality, Bayesian CR inequality, and Bayesian $\alpha$-CR inequality. These are obtained from, respectively, $I_\alpha$-divergence (or relative $\alpha$-entropy), generalized Csisz\'ar divergence, Bayesian KL divergence, and Bayesian $I_\alpha$-divergence.
\end{abstract}

\begin{keywords}
\kwd{Bayesian bounds}
\kwd{Cram\'{e}r-Rao lower bound}
\kwd{R\'enyi entropy}
\kwd{Fisher metric}
\kwd{relative $\alpha$-entropy}
\end{keywords}

%\end{frontmatter}

\section{Introduction}
\label{sec:introduction}
Information geometry   is   a   study   of   statistical   models   (families of probability distributions) from a   Riemannian   geometric   perspective. In this framework, a statistical model plays the role of a manifold. Each point on the manifold is a probability distribution from the model. In a historical development, Prof. C R Rao introduced this idea in his seminal 1945 paper \cite[Secs.~6,7]{rao1945information}. He also proposed Fisher information as a Riemannian metric on a statistical manifold as follows: Let $\mathcal{P}$ be the space of all probability distributions (strictly positive) on a state space $\mathbb{X}$. Assume that $\mathcal{P}$ is parametrized by a coordinate system $\theta$. Then, the {\em Fisher metric} at a point $p_\theta$ of $\mathcal{P}$ is 
\begin{align}
\label{eqn:fisher_matrix}
  g_{i,j}(\theta) := \langle \partial_i,\partial_j\rangle_{p_\theta} & := \int \frac{\partial}{\partial\theta_i} p_\theta(x)\cdot\frac{\partial}{\partial\theta_j}\log p_\theta(x)\, dx\\ 
  & = \label{eqn:derivative_KL-div} -\frac{\partial}{\partial\theta_i}\frac{\partial}{\partial\theta_j'}I(p_{\theta},p_{\theta'})\bigg|_{\theta=\theta'},
\end{align}
where $I(p_{\theta},p_{\theta'})$ is the Kullback-Leibler (KL)-divergence between $p_{\theta}$ and $p_{\theta'}$ (or {\em entropy of $p_{\theta}$ relative to $p_{\theta'}$}). Rao called the space based on such a metric a Riemann space and the geometry associated with this as the Riemannian geometry with its definitions of length, distance, and angle. 

Since then, information geometry has widely proliferated through several substantial contributions, for example, Efron \cite{Efron1975curvature}, Cencov \cite{cencov1981statistical}, Amari \cite{amari1982curved}, \cite{amari1985differential}, Amari and Nagoaka \cite{amari2000methods}, and Eguchi \cite{eguchi1992geometry}. Information-geometric concepts have garnered considerable interest in recent years with a wide range of books by Amari \cite{amari2016information}, Ay et al. \cite{ay2017information}, Ay et al. \cite{ay2018information}, Barndorff-Nielsen \cite{barndorff2014information}, Calin and Udri{\c{s}}te \cite{calin2014geometric}, Kass and Vos \cite{kass2011geometrical}, Murray and Rice \cite{murray2017differential}, Nielsen \cite{nielsen2021progress}, Nielsen and Bhatia \cite{nielsen2013matrix}, and Nielsen et al. \cite{nielsen2017computational}. This perspective is helpful in analyzing problems in engineering and sciences where parametric probability distributions are used, including (but not limited to) robust estimation of co-variance matrices \cite{balaji2014information}, optimization \cite{amari2013minkovskian}, signal processing \cite{amari2016information}, neural networks \cite{amari1997information,amari2002information}, machine learning \cite{amari1998natural}, optimal transport \cite{gangbo1996geometry}, quantum information \cite{grasselli2001uniqueness}, radar systems \cite{de2014design,barbaresco2008innovative}, communications \cite{coutino2016direction}, computer vision \cite{maybank2012fisher}, and covariant thermodynamics \cite{barbaresco2014koszul,barbaresco2016geometric}. More recently, several developments in deep learning \cite{desjardins2015natural,roux2008topmoumoute} that employ various approximations to the Fisher information matrix (FIM) to calculate the gradient descent have incorporated information-geometric concepts.

\begin{figure}[ht]
\includegraphics[width=0.35\textwidth]{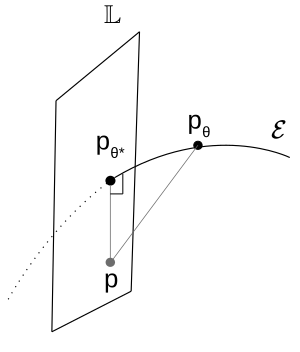}
\caption{Orthogonality of exponential and linear families}
\label{fig:E_and_L}
\end{figure}

We are aware of two strong motivations for studying information geometry. The first is the following. The pair of statistical models, namely linear and exponential families of probability distributions, play an important role in information geometry. These are {\em dually flat} in the sense that the former is flat with respect to the m-connection and the later is flat with respect to the e-connection and the two connections are dual to each other with respect to Fisher metric (see \cite[Sec.~2.3 and Ch.~3]{amari2000methods}). We refer the reader to \cite{kurose1994flat} and \cite{matsuzoe1998flat} for further details on the importance of dualistic structures in Riemannian geometry. A close relationship between the linear and exponential families were known even without Riemannian geometry. These two families were shown to be \quotes{orthogonal} to each other in the sense that an exponential family intersects with the associated linear family in a single point at right angle, that is, a Pythagorean theorem with respect to the KL-divergence holds at the point of intersection (See Fig.~\ref{fig:E_and_L}). This is interesting as it enables one to turn the problem of {\em maximum likelihood estimation} (MLE) on an exponential family into a problem of solving a set of linear equations \cite[Th.~3.3]{csiszar2004information}. This fact was extended to generalised exponential families and convex integral functionals (which includes Bregman divergences) by Csisz\'ar and Mat\'u\v{s} \cite[Sec.~4]{CsiszarM12J}. An analogous fact was shown from a Riemannian geometric perspective for U-divergences (a special form of Bregman divergences) and U-models (Student distributions are a special case) by Eguchi et al. \cite{EguchiKO14J}. A similar orthogonality relationship between power-law and linear families with respect to the $I_\alpha$-divergence (or relative $\alpha$-entropy) was established in \cite{kumar2015minimization-1}.

\begin{figure}[ht]
\includegraphics[width=0.5\textwidth]{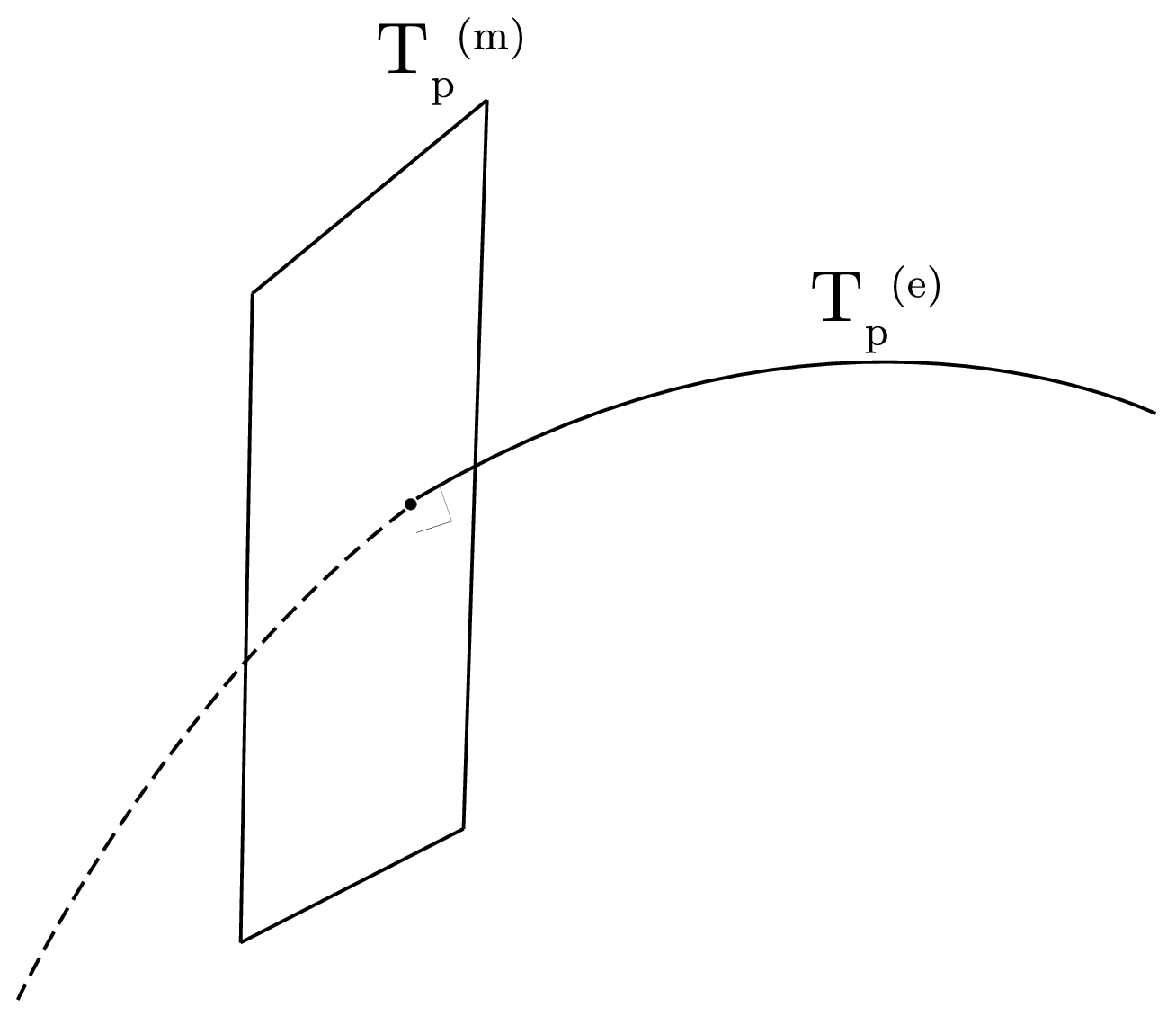}
\caption{Orthogonality of $T_p^{(m)}$ and $T_p^{(e)}$}
\label{fig:orthogonal}
\end{figure}

The second motivation for information geometry (and also for this chapter) comes from the works of Amari and Nagoaka \cite[Sec.~2.5]{amari2000methods}. Apart from showing that the $e$ and $m$ connections are dual to each other with respect to the Fisher metric, they also define, at every point $p$ of a manifold $S$, a pair of spaces of vectors $T_p^{(m)}$ and $T_p^{(e)}$ and show that $T_p^{(m)}$ is flat with respect to the $m$-connection and $T_p^{(e)}$ is flat with respect to the $e$-connection and are \quotes{orthogonal} to each other with respect to the Fisher metric (see Fig.~\ref{fig:orthogonal}). Also the Fisher metric in (1) for two tangent vectors $X$ and $Y$ can be given  by $\langle X, Y\rangle_p = \langle X^{(m)},Y^{(e)}\rangle_p$, where $X^{(m)}\in T_p^{(m)}, Y^{(e)}\in T_p^{(e)}$. They show that, for a smooth function $f:S\to\mathbb{R}$,
\begin{equation}
\label{eqn:norm_differential}
\|(df)_p\|_p^2 = (\partial_i f)_p (\partial_j f)_p g^{i,j}(p),
\end{equation}
where $g^{i,j}$ are the entries of the inverse of the FIM defined in \eqref{eqn:fisher_matrix}. This enables them to show that, for a random variable $A:\mathbb{X}\to \mathbb{R}$,
\begin{equation}
V_p[A] = \|(dE[A])_p\|_p^2.
\end{equation}
where $E[A]: \mathcal{P}\to \mathbb{R}$ maps $p\mapsto E_p[A]$, the expectation of $A$ with respect to $p$ and $V_p[A]$, the variance \cite[Th.~2.8]{amari2000methods}. This is interesting as this connects Riemannian geometry and statistics (as the left hand side is a statistical quantity and the right side is a differential geometric quantity). The above, when applied to a sub-manifold $S$ of $\mathcal{P}$, becomes
\begin{equation}
\label{eqn:var_diff_relation}
V_p[A] \ge \|(dE[A])_p\|_p^2.
\end{equation}
Now, if $\widehat{\theta}=(\widehat{\theta_1},\dots, \widehat{\theta_k})$ is an unbiased estimator of $\theta=(\theta_1,\dots, \theta_k)$ (assuming that $S$ is a $k$-dimensional manifold), then applying \eqref{eqn:var_diff_relation} to $A = \sum_i c_i\widehat{\theta_i}$ for $c=(c_1,\dots,c_k)^T$, we get the classical Cram\'er-Rao lower bound (CRLB)
\begin{equation}
    \label{eqn:cramer_Rao_classical}
    c^T V_\theta(\widehat{\theta})c\ge c^T G(\theta)^{-1} c,
\end{equation}
where $V_\theta(\widehat{\theta})$ is the covariance matrix of $\widehat{\theta}$ and $G(\theta)$ is the FIM. This is one among several ways of proving the Cram\'er-Rao (CR) inequality. This is interesting from a divergence function point of view as Fisher metric and the $e$ and $m$ connections can be derived from the KL-divergence. Indeed, Eguchi \cite{eguchi1992geometry} proved that, given a (sufficiently smooth) divergence function, one can always come up with a metric and a pair of affine conections so that this triplet forms a dualistic structure on the underlying statistical manifold. In this chapter, we first apply Eguchi's theory to the $I_\alpha$-divergence and come up with a dualistic structure of a metric and a pair of affine connections. Subsequently, we apply Amari and Nagoaka's above mentioned theory to establish an $\alpha$-version of the Cramer-Rao inequality. We then extend this to generalised Csisz\'ar divergences and obtain a generalised Cramer-Rao inequality. We also establish the Bayesian counterparts of the $\alpha$-Cramer-Rao inequality and the usual one by defining the appropriate divergence functions. 
    
\section{$I$-divergence and $I_\alpha$-divergence}
\label{subsec:intro_relent}
In this section, we introduce $I_\alpha$-divergence and its connection to Csisz\'ar divergences. We restrict ourselves to finite state space $\X$. However, all these may be extended to continuous densities using analogous functional analytic tools (see our remark on infinite $\X$ in subsection \ref{subsec:infinite}). 

The $I$-divergence between two probability distributions $p$ and $q$ on a finite state space, say $\X = \{0,1,2,\dots,M\}$, is defined as
\begin{align}
\label{eqn:rel_ent}
    {I}(p,q) := \sum_{x\in\mathbb{X}} p(x)\log p(x) - \sum_{x\in\mathbb{X}} p(x) \log q(x),
\end{align}
where 
\begin{align}
\label{eq:shannon_ent}
    H(p) := -\sum_{x\in\mathbb{X}} p(x)\log p(x)
\end{align}
is the {\em Shannon entropy} and
\begin{align}
\label{eq:shannon_cross_ent}
    D(p\|q) := -\sum_{x\in\mathbb{X}} p(x)\log q(x)
\end{align}
is the {\em cross-entropy}. Throughout the chapter, we shall assume that all probability distributions have common support $\mathbb{X}$.
  
There are other measures of uncertainty that are used as alternatives to Shannon entropy. One of these is the R\'{e}nyi entropy that was discovered by Alfred R\'enyi while attempting to find an axiomatic characterization to measures of uncertainty \cite{renyi1961measures}. Later, Campbell gave an operational meaning to R\'enyi entropy  \cite{campbell1965coding}; he showed that R\'enyi entropy plays the role of Shannon entropy in a source coding problem where normalized cumulants of compressed lengths are considered. Blumer and McEliece \cite{blumer1988renyi} and Sundaresan \cite{sundaresan2007guessing} studied the mismatched (source distribution) version of this problem and showed that $I_\alpha$-divergence plays the role of $I$-divergence in this problem. The R\'{e}nyi entropy of $p$ of order $\alpha$, $\alpha \ge 0$, $\alpha\neq 1$, is defined as 
\[
H_{\alpha}(p) := \frac{1}{1-\alpha}\log\sum_x p(x)^{\alpha}.
\]
{\em $I_\alpha$-divergence} (also known as {\em Sundaresan's divergence} \cite{200206ISIT_Sun}) between two probability distributions $p$ and $q$ is defined as

%\par\noindent\small
\begin{align}
\label{eqn:alphadiv_compress}
\lefteqn{I_{\alpha}(p,q)}\nonumber\\
& := \frac{1}{1-\alpha}\log \sum_x p(x)\left(\frac{q(x)}{\|q\|_{\alpha}}\right)^{\alpha-1}  - \frac{1}{\alpha(1-\alpha)}\log \sum_x p(x)^{\alpha}\\\label{eqn:alphadiv_alt}
 & = \frac{1}{1-\alpha} \log \sum_x p(x) q(x)^{\alpha-1} + \frac{1}{\alpha}\log \sum_x q(x)^{\alpha} - \frac{1}{\alpha(1-\alpha)}\log \sum_x p(x)^{\alpha}.\nonumber\\
\end{align}

%\normalsize

The first term in (\ref{eqn:alphadiv_compress}) is called the {\em Renyi cross-entropy} and is to be compared with the first term of (\ref{eqn:rel_ent}). It should be noted that, as $\alpha \rightarrow 1$, we have $I_{\alpha}(p,q) \rightarrow I(p,q)$ and $H_{\alpha}(p) \rightarrow H(p)$ \cite{kumar2015minimization-1}. R\'enyi entropy and $I_\alpha$-divergence are related by the equation $I_\alpha(p,u) = \log |\mathbb{X}| - H_\alpha(p)$. 
   
The ubiquity of R\'enyi entropy and $I_\alpha$-divergence in information theory was further noticed, for example, in guessing problems by Ar{\i}kan \cite{arikan1996inequality}, Sundaresan \cite{sundaresan2007guessing}, and Huleihel et al. \cite{huleihel2017guessing}; and in encoding of tasks by Bunte and Lapidoth \cite{bunte2014codes}. $I_\alpha$-divergence arises in statistics as a generalized likelihood function robust to outliers \cite{jones2001comparison}, \cite{kumar2015minimization-2}. It has been referred variously as $\gamma$-divergence \cite{fujisawa2008robust,cichocki2010families,notsu2014spontaneous}, projective power divergence \cite{eguchi2011projective,eguchi2010entropy}, logarithmic density power divergence \cite{basu2011statistical} and relative $\alpha$-entropy \cite{200206ISIT_Sun}, \cite{kumar2015minimization-1}. Throughout this chapter, we shall follow the nomenclature of $I_\alpha$-divergence.
  
$I_\alpha$-divergence shares many interesting properties with $I$-divergence (see, e.g. \cite[Sec.~II]{kumar2015minimization-1} for a summary of its properties and relationships to other divergences). For instance, analogous to $I$-divergence, $I_\alpha$-divergence behaves like squared Euclidean distance and satisfies a Pythagorean property \cite{kumar2015minimization-1,kumar2018information}. The Pythagorean property proved useful in arriving at a computation scheme \cite{kumar2015minimization-2} for a robust estimation procedure \cite{fujisawa2008robust}.

\subsection{Extension to Infinite $\mathbb{X}$}
\label{subsec:infinite}
The Cramer-Rao type inequalities discussed in this chapter are obtained by applying Eguchi's theory \cite{eguchi1992geometry} followed by Amari-Nagaoka's framework \cite[Sec.~2.5]{amari2000methods}. While the former is applicable even for infinite $\mathbb{X}$, the latter \cite[Sec.~2.5]{amari2000methods} is applicable only for the finite case. This is a limitation on the applicability of the established bounds. Several works, notably Pistone  \cite{Pistone1995Annals,Pistone2007Annals} have made significant contributions in this direction; see also \cite{Amari2021Information}, \cite{ay2017information} for further details. A more interesting case from the applications perspective is when $\mathbb{X}$ is infinite and $S$ is finite-dimensional. It follows from the concluding remarks of Amari \cite[Sec.~2.5]{amari2000methods} and via personal communication (dated 29 June 2021) with Prof. Nagaoka that the arguments of \cite[Sec.~2.5]{amari2000methods} would still ``apply in its essence''. However, the formulation of these arguments in a mathematically rigorous way in the framework of infinite-dimensional differential geometry on $\mathcal{P}(\mathbb{X})$ is worth investigating. %This will be the focus in one of our forthcoming works.

\subsection{Bregman vs Csisz\'{a}r}
\label{sec:bregman}
Bregman and Csisz\'ar are two classes of divergences with the $I$-divergence at their intersection. Our primary interest in this chapter is the geometry of $I_\alpha$-divergence. This divergence differs from, but is related to, the usual R\'enyi divergence which is a member of Csisz\'ar family. However, $I_\alpha$-divergence is not a member of the Csisz\'ar family. Instead, it falls under a \textit{generalised} form of Csisz\'ar $f$-divergences, whose %So we investigate the geometry of \textit{generalized} Csisz\'ar divergences,
geometry is different from that of Bregman and Csisz\'ar divergences  \cite{zhang2004divergence}. In particular, $I_\alpha$-divergence is closely related to the Csisz\'ar $f$-divergence $D_{f}$ as
\begin{equation}
\label{eqn:the-function-in_D_f}
    I_{\alpha}(p,q) = \frac{1}{1-\alpha} \log\left[ \text{sgn}(1-\alpha) \cdot D_{f}(p^{(\alpha)},q^{(\alpha)}) + 1\right],
\end{equation}
where \par\noindent\small
\begin{equation*}
    p^{(\alpha)}(x) := \frac{p(x)^{\alpha}}{\sum_y {p(y)}^{\alpha}}, q^{(\alpha)}(x) := \frac{q(x)^{\alpha}}{\sum_y {q(y)}^{\alpha}}, f(u) = \text{sgn}(1-\alpha) \cdot (u^{{1}/{\alpha}} - 1), u \geq 0
\end{equation*}\normalsize
[c.f. \cite[Sec.~II]{kumar2015minimization-1}]. The measures $p^{(\alpha)}$ and $q^{(\alpha)}$ are called {\em $\alpha$-escort} or {\em $\alpha$-scaled} measures \cite{tsallis1998role}, \cite{karthik2018on}. Observe from \eqref{eqn:the-function-in_D_f} that $I_\alpha$-divergence is a monotone function of the Csisz\'ar divergence, not between $p$ and $q$, but their escorts $p^{(\alpha)}$ and $q^{(\alpha)}$. For a strictly convex function $f$ with $f(1) = 0$, the Csisz\'ar $f$-divergence between two probability distributions $p$ and $q$ is defined as (also, see \cite{csiszar1991why})
\begin{equation*}
  D_f(p,q) = \sum_x q(x) f\left(\frac{p(x)}{q(x)}\right).
\end{equation*}
Note that the right side of (\ref{eqn:the-function-in_D_f}) is R\'enyi divergence between $p^{(\alpha)}$ and $q^{(\alpha)}$ of order ${1}/{\alpha}$ \cite{kumar2015minimization-2}. For an extensive study of properties of the R\'enyi divergence, we refer the reader to \cite{vanerven2014renyi}. The Csisz\'ar $f$-divergence is further related to the Bregman divergence $B_f$ through
\begin{align}
    D_f(p, q) = \sum_x p(x) B_f({q(x)}/{p(x)},1),
\end{align}
 \cite{zhang2004divergence}. $I_\alpha$-divergence differs from both Csisz\'ar and Bregman divergences because of the appearance of the escort distributions in (\ref{eqn:the-function-in_D_f}).

\subsection{Classical vs Quantum CR inequality}
\label{sec:qcrlb}
This chapter is concerned with the \underline{classical} CR inequality to differentiate it with its \underline{quantum} counterpart. In quantum metrology, the choice of measurement affects the probability distribution obtained. The implication of this effect is that the classical FIM becomes a function of measurement. In general, there may not be any measurement to attain the resulting quantum FIM \cite{braunstein1994statistical}. There are many quantum versions of classical FIM, e.g. based on the symmetric, left, and right derivatives. Petz \cite{petz1996monotone,petz2007quantum} showed that all quantum FIMs are a member of a family of Riemannian monotone metrics. Further, all quantum FIMs yield quantum CR inequalities with different achievabilities \cite{liu2019quantum}. Quantum algorithms to estimate von Neumann's entropy and $\alpha$-R\'{e}nyi entropy of quantum states (with Hartley, Shannon, and collision entropies as special cases for $\alpha=0$, $\alpha=1$, and $\alpha=2$, respectively) have also been reported  \cite{li2018quantum}. For geometric structure induced from a quantum divergence, we refer the reader to \cite[Chapter 7]{amari2000methods}.

\section{Information Geometry from a Divergence Function}
\label{sec:desiderata}
In this section, we summarize the information-geometric concepts associated with a general divergence function. For detailed mathematical definitions, we refer the reader to Amari and Nagoaka \cite{amari2000methods}. For more intuitive explanations of information-geometric notions, one may refer to Amari's recent book \cite{amari2016information}. We shall introduce the reader to a certain dualistic structure on a statistical manifold of probability distributions arising from a divergence function. For a detailed background on differential and Riemannian geometry, we refer the reader to \cite{spivak2005comprehensive,jost2005riemannian,gallot2004riemannian,docarmo1976differential}.
  
In information geometry, statistical models play the role of a manifold and the FIM and its various generalizations play the role of a Riemannian metric. A {\em statistical manifold} is a parametric family of probability distributions on $\mathbb{X}$ with a \quotes{continuously varying} parameter space $\Theta$ (statistical model). A statistical manifold $S$ is usually represented by $S = \{p_{\theta}: \theta = (\theta_1,\dots,\theta_n)\in \Theta\subset\mathbb{R}^n\}$. Here, $\theta_1,\dots,\theta_n$ are the coordinates of the point $p$ in $S$ and the mapping $p \mapsto (\theta_1(p),\dots,\theta_n(p))$ that takes a point $p$ to its coordinates constitute a {\em coordinate system}. The \quotes{dimension} of the parameter space is the dimension of the manifold. For example, the set of all binomial probability distributions $\{B(r,\theta):\theta\in (0,1)\}$, where $r$ is the (known) number of trials, is a one-dimensional statistical manifold. Similarly, the family of normal distributions $S = \{N(\mu,\sigma^2) : \mu\in \mathbb{R}, \sigma^2 > 0\}$ is a two dimensional statistical manifold. The {\em tangent space} at a point $p$ on a manifold $S$ (denoted $T_p(S)$) is a linear space that corresponds to the \quotes{local linearization} of the manifold around the point $p$. The elements of $T_p(S)$ are called {\em tangent vectors} of $S$ at $p$. For a coordinate system $\theta$, the (standard) basis vectors of a tangent space $T_p$ are denoted by $(\partial_i)_p := \left({\partial}/{\partial\theta_i}\right)_p, i=1,\dots,n$. A {\em (Riemannian) metric} at a point $p$ is an inner product defined for any pair of tangent vectors of $S$ at $p$. 

A metric is completely characterized by the matrix whose entries are the inner products between the basic tangent vectors. That is, it is characterized by the matrix
  \[
  G(\theta) = [g_{i,j}(\theta)]_{i,j = 1,\dots,n},
  \]
  where $g_{i,j}(\theta) := \langle \partial_i, \partial_j\rangle $. An {\em affine connection} (denoted $\nabla$) on a manifold is a correspondence between the tangent vectors at a point $p$ to the tangent vectors at a \quotes{nearby} point $p'$ on the manifold. An affine connection is completely specified by specifying the $n^3$ real numbers $(\Gamma_{ij,k})_p, i,j,k=1,\dots,n$ called the {\em connection coefficients} associated with a coordinate system $\theta$.

Let us restrict to statistical manifolds defined on a finite set $\mathbb{X} = \{a_1,\dots,a_d\}$.  Let $\mathcal{P} := \mathcal{P}(\mathbb{X})$ denote the space of all probability distributions on $\mathbb{X}$. Let $S\subset\mathcal{P}$ be a sub-manifold. Let $\theta = (\theta_1,\dots,\theta_k)$ be a parameterization of $S$. Let $D$ be a divergence function on $S$. By a {\em divergence}, we mean a non-negative function $D$ defined on $S\times S$ such that $D(p,q) = 0$ iff $p=q$.  Let $D^*$ be another divergence function defined by $D^*(p,q) = D(q,p)$. Given a (sufficiently smooth) divergence function on $S$, Eguchi \cite{eguchi1992geometry} defines a Riemannian metric on $S$ by the matrix
\[
G^{(D)}(\theta) = \left[g_{i,j}^{(D)}(\theta)\right],
\]
where
\begin{align*}
  g_{i,j}^{(D)}(\theta) := -D[\partial_i,\partial_j] :=  -\frac{\partial}{\partial\theta_j'}\frac{\partial}{\partial\theta_i}D(p_{\theta},p_{\theta'})\bigg|_{\theta=\theta'}
\end{align*}
where $g_{i,j}$ is the elements in the $i$th row and $j$th column of the matrix $G$, $\theta = (\theta_1,\dots,\theta_n)$, $\theta' = (\theta_1',\dots,\theta_n')$, and dual affine connections $\nabla^{(D)}$ and $\nabla^{(D^*)}$, with connection coefficients described by following Christoffel symbols
\begin{align*}
  \Gamma_{ij,k}^{(D)}(\theta) := -D[\partial_i\partial_j,\partial_k]
  := -\frac{\partial}{\partial\theta_i}\frac{\partial}{\partial\theta_j}\frac{\partial}{\partial\theta_k'}D(p_{\theta},p_{\theta'})\bigg|_{\theta=\theta'}
\end{align*}
and
\begin{align*}
  \Gamma_{ij,k}^{(D^*)}(\theta) & := & -D[\partial_k,\partial_i\partial_j] :=  -\frac{\partial}{\partial\theta_k}\frac{\partial}{\partial\theta_i'}\frac{\partial}{\partial\theta_j'}D(p_{\theta},p_{\theta'})\bigg|_{\theta=\theta'},
\end{align*}
such that $\nabla^{(D)}$ and $\nabla^{(D^*)}$ are duals of each other with respect to the metric $G^{(D)}$ in the sense that
\begin{align}
  \label{dualistic-structure}
  \partial_k g_{i,j}^{(D)}=\Gamma_{ki,j}^{(D)}+ \Gamma_{kj,i}^{(D^*)}.
\end{align}

When $D(p,q) = I(p,q)$, the resulting metric is called the {\em Fisher information metric} given by $G(\theta) = [g_{i,j}(\theta)]$ with
  \begin{align}
  \nonumber
  g_{i,j}(\theta)
  & = \left. - \frac{\partial}{\partial \theta_i} \frac{\partial}{\partial \theta'_j} \sum_{x} p_{\theta}(x) \log \frac{p_{\theta}(x)}{p_{\theta'}(x)} \right|_{\theta' = \theta}\\
  \nonumber
  & = \sum_x \partial_i p_{\theta}(x)\cdot \partial_j \log p _{\theta}(x)\\
  \nonumber
  & = E_{\theta}[\partial_i \log p_{\theta}(X)\cdot \partial_j \log p_{\theta}(X)]\\
  \label{eqn:fisher-information-metric-deterministic}
  & = \text{Cov}_{\theta}[\partial_i \log p_{\theta}(X), \partial_j \log p_{\theta}(X)].
  \end{align}
  The last equality follows from the fact that the expectation of the score function is zero, that is, $E_{\theta}[\partial_i \log p_{\theta}(X)] = 0, i = 1, \dots, n$. The affine connection $\nabla^{(I)}$ is called the {\em $m$-connection} (mixture connection) with connection coefficients %given by
  \begin{align*}
    \Gamma_{ij,k}^{(m)}(\theta)
    = \sum_x \partial_i \partial_j p_{\theta}(x)\cdot \partial_k \log p _{\theta}(x)
  \end{align*}
   and is denoted $\nabla^{(m)}$. The affine connection $\nabla^{(I^*)}$ is called the {\em $e$-connection} (exponential connection) with connection coefficients
  \begin{align*}
    \Gamma_{ij,k}^{(e)}(\theta)
    = \sum_x \partial_k p _{\theta}(x)\cdot \partial_i \partial_j \log p _{\theta}(x)
  \end{align*}
  and is denoted $\nabla^{(e)}$   \cite[Sec. 3.2]{amari2000methods}).

\subsection{Information Geometry for \texorpdfstring{$\alpha$}{}-CR inequality}
\label{subsec:alphaCRLB}  
  Set $D = I_{\alpha}$ and apply the Eguchi framework. For simplicity, write $G^{(\alpha)}$ for $G^{(I_{\alpha})}$. The Riemannian metric on $S$ is specified by the matrix $G^{(\alpha)}(\theta) = [g_{i,j}^{(\alpha)}(\theta)]$, where\par\noindent\small
  \begin{align}
  \lefteqn{g_{i,j}^{(\alpha)}(\theta) ~ := ~ g_{i,j}^{(I_{\alpha})} } \nonumber\\
  & = -\frac{\partial}{\partial\theta_j'}\frac{\partial}{\partial\theta_i}I_{\alpha}(p_{\theta},p_{\theta'})\bigg|_{\theta' = \theta} \nonumber\\\label{eqn:alpha-metric}\\
    & = \frac{1}{\alpha-1}\cdot\frac{\partial}{\partial\theta_j'}\frac{\partial}{\partial\theta_i} \left[\log \sum_y p_{\theta}(x) {p_{\theta'}(x)}^{\alpha-1}\right]_{\theta' = \theta}\\
  & = \frac{1}{\alpha-1}\sum_x \partial_i p_{\theta}(x)\cdot \partial_j'\left[\frac{{p_{\theta'}(x)}^{\alpha-1}}{\sum_y p_{\theta}(y) {p_{\theta'}(y)}^{\alpha-1}}\right]_{\theta' = \theta}\\
    & = \sum_x \partial_i p_{\theta}(x)\left[\frac{{p_{\theta}(x)}^{\alpha-2}\partial_j p_{\theta}(x) \sum_y p_{\theta}(y)^{\alpha} - p_{\theta}(x)^{\alpha-1}\sum_y p_{\theta}(y)^{\alpha-1}\partial_j p_\theta(y)}{(\sum_y p_{\theta}(y)^{\alpha})^2}\right]\\
  & = E_{\theta^{(\alpha)}}[\partial_i (\log p_{\theta}(X))\cdot \partial_j (\log p_{\theta}(X))]\nonumber\\
  \label{eqn:g-alpha-expansion}
  & \hspace{1.5cm}-E_{\theta^{(\alpha)}}[\partial_i \log p_{\theta}(X)]\cdot E_{\theta^{(\alpha)}}[\partial_j \log p_{\theta}(X)]\\\label{eqn:alpha-metric-covariance}
  & = \text{Cov}_{\theta^{(\alpha)}}[\partial_i \log p_{\theta}(X), \partial_j \log p_{\theta}(X)]\\
  \label{eqn:RiemannianOnS-alpha}
  & = \frac{1}{\alpha^2}\text{Cov}_{\theta^{(\alpha)}}[\partial_i \log p_{\theta}^{(\alpha)}(X), \partial_j \log p_{\theta}^{(\alpha)}(X)],
  \end{align}\normalsize
  where $p_{\theta}^{(\alpha)}$ is the $\alpha$-escort distribution associated with $p_{\theta}$,
  \begin{equation}
  \label{eqn:escort_distribution}
  p_{\theta}^{(\alpha)}(x) := \frac{p_{\theta}(x)^{\alpha}}{\sum_y {p_{\theta}(y)}^{\alpha}},
  \end{equation}
  and $E_{\theta^{(\alpha)}}$ denotes expectation with respect to $p_{\theta}^{(\alpha)}$. The equality (\ref{eqn:RiemannianOnS-alpha}) follows because\par\noindent\small
  \begin{align*}
  \partial_i p_\theta^{(\alpha)}(x) = \partial_i\left(\frac{p_\theta(x)^\alpha}{\sum_y p_\theta(y)^\alpha}\right) = \alpha\left[\frac{{p_{\theta}^{(\alpha)}(x)}}{p_{\theta}(x)}\partial_i p_{\theta}(x) - p_{\theta}^{(\alpha)}(x) \sum_y \frac{{p_{\theta}^{(\alpha)}(y)}}{p_{\theta}(y)}\partial_i p_{\theta}(y)\right].
  \end{align*}\normalsize
  
  If we define $S^{(\alpha)} := \{p_\theta^{(\alpha)} : p_\theta\in S\}$, then (\ref{eqn:RiemannianOnS-alpha}) tells us that $G^{(\alpha)}$ is essentially the usual Fisher information for the model $S^{(\alpha)}$ up to the scale factor $\alpha$.
  
  We shall call the metric defined by $G^{(\alpha)}$ an \emph{$\alpha$-information metric}. We shall assume that $G^{(\alpha)}$ is positive definite; see \cite[pp. 39-40]{kumar2020cram} for an example of a parameterization with respect to which this assumption holds.
   
  Let us now return to the general manifold $S$ with a coordinate system $\theta$. Denote $\nabla^{(\alpha)}:=\nabla^{(I_{\alpha})}$ and $\nabla^{(\alpha)*}:=\nabla^{(I_{\alpha}^*)}$ where the right-hand sides are as defined by Eguchi \cite{eguchi1992geometry} with $D = I_{\alpha}$.
  
  Motivated by the expression for the Riemannian metric in (\ref{eqn:alpha-metric}), define
  \begin{equation}
  \label{alpha-partial}
  \partial_i^{(\alpha)}(p_\theta (x)) := \frac{1}{\alpha-1}\partial_i'\left(\frac{{p_{\theta'}(x)}^{\alpha-1}}{\sum_y p_{\theta}(y)\, {p_{\theta'}(y)}^{\alpha-1}}\right)\bigg |_{\theta'=\theta}.
  \end{equation}
  We now identify the corresponding connection coefficients as
  \begin{align}
  \label{eqn:connection_coefficients}
  \Gamma_{ij,k}^{(\alpha)}& := \Gamma_{ij,k}^{(I_{\alpha})}\\
  & = - I_{\alpha}[\partial_i\partial_j,\partial_k] \nonumber\\
  & = \frac{1}{\alpha-1} \left[\sum_x \partial_j p_{\theta}(x) \cdot \partial_i\left(\partial_k^{(\alpha)}(p_\theta)\right) + \sum_x \partial_i\partial_j p_{\theta}(x) \cdot \partial_k^{(\alpha)}(p_\theta)\right]
  \end{align}
  and
  \begin{align}
  \label{eqn:dual_connection_coefficients}
  \Gamma_{ij,k}^{(\alpha)*} & := \Gamma_{ij,k}^{(I_{\alpha}^*)}\\ 
  & = -I_{\alpha}[\partial_k,\partial_i\partial_j] \nonumber\\
  & = \frac{1}{\alpha-1}\left[\sum_x \partial_k p_{\theta}(x) \cdot \partial_i'\partial_j'\left(\frac{{p_{\theta'}(x)}^{\alpha-1}}{\sum_y p_{\theta}(y){p_{\theta'}(y)}^{\alpha-1}}\right)\bigg |_{\theta'=\theta}\right].\nonumber\\
  \end{align}
  We also have (\ref{dualistic-structure}) specialized to our setting:
  \begin{align}
  \label{eqn:dual_connection}
  \partial_k g_{i,j}^{(\alpha)}=\Gamma_{ki,j}^{(\alpha)}+ \Gamma_{kj,i}^{(\alpha)*}.
  \end{align}
  $(G^{(\alpha)},\nabla^{(\alpha)},\nabla^{(\alpha)*})$ forms a dualistic structure on $S$. We shall call the connection $\nabla^{(\alpha)}$ with the connection coefficients $\Gamma_{ij,k}^{(\alpha)}$, an \emph{$\alpha$-connection}.
  
  When $\alpha = 1$, the metric $G^{(\alpha)}(\theta)$ coincides with the usual Fisher metric and the connections $\nabla^{(\alpha)}$ and $\nabla^{(\alpha)*}$ coincide with the $m$-connection $\nabla^{(m)}$ and the $e$-connection $\nabla^{(e)}$, respectively.
  
  A comparison of the expressions in (\ref{eqn:fisher-information-metric-deterministic}) and (\ref{eqn:RiemannianOnS-alpha}) suggests that the manifold $S$ with the $\alpha$-information metric may be equivalent to the Riemannian metric specified by the FIM on the manifold $S^{(\alpha)} := \{ p_{\theta}^{(\alpha)} : \theta \in \Theta \subset \mathbb{R}^n \}$. This is true to some extent because the Riemannian metric on $S^{(\alpha)}$ specified by the FIM is simply $G^{(\alpha)}(\theta) = [g_{ij}^{(\alpha)}(\theta)]$. However, our calculations indicate that the $\alpha$-connection and its dual on $S$ are not the same as the $e$- and the $m$-connections on $S^{(\alpha)}$ except when $\alpha = 1$. The $\alpha$-connection and its dual should therefore be thought of as a parametric generalization of the $e$- and $m$-connections. In addition, the $\alpha$-connections in (\ref{eqn:connection_coefficients}) and (\ref{eqn:dual_connection_coefficients}) are different from the $\alpha$-connection of Amari and Nagaoka \cite{amari2000methods}, which is a convex combination of the $e$- and $m$-connections.
  
  \subsection{An \texorpdfstring{$\alpha$}{}-Version of Cram\'{e}r-Rao Inequality}
  \label{subsec:analogous_cr_inequality}
  We now apply Amari and Nagoaka's theory \cite[2.5]{amari2000methods} to derive the $\alpha$-CR inequality. For this, we examine the geometry of $\mathcal{P}$ with respect to the metric $G^{(\alpha)}$ and the dual affine connections $\nabla^{(\alpha)}$ and $\nabla^{(\alpha)^*}$. %Later, we formulate an $\alpha$-equivalent version of the Cram\'{e}r-Rao inequality associated with a submanifold $S$.
  Note that $\mathcal{P}$ is an open subset of the affine subspace $\mathcal{A}_1 := \{A\in \mathbb{R}^{\mathbb{X}}:\sum\limits_x A(x) = 1\}$ and the tangent space at each $p \in \mathcal{P}$, $T_p(\mathcal{P})$ is the linear space
  \[
  \mathcal{A}_0 := \{A\in \mathbb{R}^{\mathbb{X}}:\sum\limits_x A(x) = 0\}.
  \]
  For every tangent vector $X\in T_p(\mathcal{P})$, let $X_p^{(e)}(x) := X(x)/p(x)$ at $p$ and call it the \emph{exponential representation of $X$ at $p$}. The collection of exponential representations is then
  \begin{align}
  \label{exponential_tangent_space}
  T_p^{(e)}(\mathcal{P}) := \{X_p^{(e)}:X\in T_p(\mathcal{P})\} = \{A\in \mathbb{R}^{\mathbb{X}}:E_p[A]=0\},\nonumber\hspace*{-1cm}\\
  \end{align}
  where the last equality is easy to check. Observe that (\ref{alpha-partial}) is
  \begin{align}
  \label{eqn:alpha_representation}
  \partial_i^{(\alpha)}(p_\theta (x)) & = \frac{1}{\alpha-1}\partial_i'\left(\frac{{p_{\theta'}(x)}^{\alpha-1}}{\sum_y p_{\theta}(y) {p_{\theta'}(y)}^{\alpha-1}}\right)\bigg |_{\theta'=\theta}\nonumber\\
  & = \left[\frac{{p_{\theta}(x)}^{\alpha-2}~\partial_i p_{\theta}(x)}{\sum_y {p_{\theta}(y)}^{\alpha}} - \frac{{p_{\theta}(x)}^{\alpha-1}~\sum_y {p_{\theta}(y)}^{\alpha-1}\partial_i p_{\theta}(y)}{(\sum_y {p_{\theta}(y)}^{\alpha})^2}\right]\nonumber\\
  & = \left[\frac{{p_{\theta}(x)}^{(\alpha)}}{p_{\theta}(x)}\partial_i(\log p_{\theta}(x)) - \frac{{p_{\theta}(x)}^{(\alpha)}}{p_{\theta}(x)}E_{\theta^{(\alpha)}}[\partial_i(\log p_{\theta}(X))]\right].
  \end{align}
  Define the above as an \emph{$\alpha$-representation of $\partial_i$ at $p_\theta$}. With this notation, the $\alpha$-information metric is 
  \begin{equation*}
  g_{i,j}^{(\alpha)}(\theta) = \sum_x \partial_i p_{\theta}(x) \cdot \partial_j^{(\alpha)}(p_\theta(x)).
  \end{equation*}
  It should be noted that $E_{\theta}[\partial_i^{(\alpha)}(p_\theta(X))] = 0$. This follows since 
  \[
  \partial_i^{(\alpha)} (p_{\theta}) = \frac{p_{\theta}^{(\alpha)}}{p_{\theta}} \partial_i \log p_{\theta}^{(\alpha)}.
  \]
  When $\alpha =1$, the right hand side of (\ref{eqn:alpha_representation}) reduces to $\partial_i(\log p_{\theta})$.
  
  Motivated by (\ref{eqn:alpha_representation}), the \emph{$\alpha$-representation of a tangent vector $X$ at $p$} is 
  \begin{align}
  \label{eqn:alpha_rep_tgt_vec}
  X_p^{(\alpha)}(x)
  & := \left[\frac{p^{(\alpha)}(x)}{p(x)}X_p^{(e)}(x) - \frac{p^{(\alpha)}(x)}{p(x)}E_{p^{(\alpha)}}[X_p^{(e)}]\right]\nonumber\\
  & = \left[\frac{p^{(\alpha)}(x)}{p(x)}\left(X_p^{(e)}(x) - E_{p^{(\alpha)}}[X_p^{(e)}]\right)\right].
  \end{align}
  The collection of all such $\alpha$-representations is
  \begin{align}
  T_p^{(\alpha)}(\mathcal{P}) := \{X_p^{(\alpha)} : X\in T_p(\mathcal{P})\}.
  \end{align}
  Clearly $E_p[X_p^{(\alpha)}] = 0$. Also, since any $A\in \mathbb{R}^{\mathbb{X}}$ with $E_p[A]=0$ is 
  \begin{align*}
    A = \left[\frac{p^{(\alpha)}}{p}\left(B-E_{p^{(\alpha)}}[B]\right)\right]
  \end{align*}
  with $B = \tilde{B}-E_p[\tilde{B}],$ where
  \[
  \tilde{B}(x) := \left[\frac{p(x)}{p^{(\alpha)}(x)} A(x)\right].
  \]
  In view of (\ref{exponential_tangent_space}), we have
  \begin{align}
  \label{e_space_equalto_alpha_space}
  T_p^{(e)}(\mathcal{P}) = T_p^{(\alpha)}(\mathcal{P}).
  \end{align}
  Now the inner product between any two tangent vectors $X,Y\in T_p(\mathcal{P})$ defined by the $\alpha$-information metric in (\ref{eqn:alpha-metric}) is 
  \begin{align}
  \label{eqn:alpha_metric_general}
  \langle X,Y\rangle^{(\alpha)}_p := E_p[X^{(e)}Y^{(\alpha)}].
  \end{align}
  Consider now an $n$-dimensional statistical manifold $S$, a submanifold of $\mathcal{P}$, together with the metric $G^{(\alpha)}$ as in (\ref{eqn:alpha_metric_general}). Let $T_p^*(S)$ be the dual space (cotangent space) of the tangent space $T_p(S)$ and let us consider for each $Y\in T_p(S)$, the element $\omega_Y\in T_p^*(S)$ which maps $X$ to $\langle X,Y\rangle^{(\alpha)}$.  The correspondence $Y\mapsto \omega_Y$ is a linear map between $T_p(S)$ and $T_p^*(S)$. An inner product and a norm on $T_p^*(S)$ are naturally inherited from $T_p(S)$ by
  \[
  \langle \omega_X,\omega_Y\rangle_p := \langle X,Y\rangle^{(\alpha)}_p
  \]
  and
  \[
  \|\omega_X\|_p := \|X\|_p^{(\alpha)} = \sqrt{\langle X,X\rangle^{(\alpha)}_p}.
  \]
  Now, for a (smooth) real function $f$ on  $S$, the \emph{differential} of $f$ at $p$, $(\text{d}f)_p$, is a member of $T_p^*(S)$ which maps $X$ to $X(f)$. The \emph{gradient of $f$ at p} is the tangent vector corresponding to $(\text{d}f)_p$, hence, satisfies
  \begin{align}
  \label{eqn:differential_of_function_alpha}
  (\text{d}f)_p(X) = X(f) = \langle (\text{grad} f)_p,X\rangle_p^{(\alpha)},
  \end{align}
  and
  \begin{align}
  \label{eqn:norm_of_differential_alpha}
  \|(\text{d}f)_p\|_p^2 = \langle (\text{grad}f)_p,(\text{grad}f)_p\rangle_p^{(\alpha)}.
  \end{align}
  Since $\text{grad}f$ is a tangent vector, 
  \begin{equation}
  \label{eqn:grad-f-alpha}
  \text{grad}f = \sum\limits_{i=1}^n h_i \partial_i
  \end{equation}
  for some scalars $h_i$. Applying (\ref{eqn:differential_of_function_alpha}) with $X = \partial_j$, for each $j=1,\dots,n$, and using (\ref{eqn:grad-f-alpha}), we obtain
  \begin{align*}
    (\partial_j)(f)
    & = \left\langle \sum\limits_{i=1}^n h_i \partial_i, \partial_j\right\rangle^{(\alpha)}\\
    & = \sum\limits_{i=1}^n h_i \langle \partial_i, \partial_j\rangle^{(\alpha)}\\
    & = \sum\limits_{i=1}^n h_i g_{i,j}^{(\alpha)}, \quad j = 1, \dots, n.
  \end{align*}
  This yields
  \[
  [h_1,\dots,h_n]^T = \left[G^{(\alpha)}\right]^{-1}[\partial_1(f),\dots,\partial_n(f)]^T,
  \]
  and so
  \begin{equation}
  \label{eqn:grad-coeff-equation-alpha}
  \text{grad}f = \sum\limits_{i,j} (g^{i,j})^{(\alpha)}\partial_j(f) \partial_i.
  \end{equation}
  From (\ref{eqn:differential_of_function_alpha}), (\ref{eqn:norm_of_differential_alpha}), and (\ref{eqn:grad-coeff-equation-alpha}), we get
  \begin{align}
  \label{differential_and_metric_alpha}
  \|(\text{d}f)_p\|_p^2 = \sum\limits_{i,j} (g^{i,j})^{(\alpha)}\partial_j(f) \partial_i(f)
  \end{align}
  where $(g^{i,j})^{(\alpha)}$ is the $(i,j)$th entry of the inverse of $G^{(\alpha)}$.
  
  With these preliminaries, we state results analogous to those in \cite[Sec.~2.5]{amari2000methods}.
  \begin{theorem}[\cite{kumar2020cram}]
  \label{thm:variance_and_norm_of_differential}
    Let $A:\mathbb{X}\to\mathbb{R}$ be any mapping (that is, a vector in $\mathbb{R}^{\mathbb{X}}$. Let $E[A]:\mathcal{P}\to \mathbb{R}$ be the mapping $p\mapsto E_p[A]$. We then have
    \begin{align}
    \label{eqn:variance_and_norm_of_differential}
    \text{Var}_{p^{(\alpha)}}\left[\frac{p}{p^{(\alpha)}}(A-E_p[A])\right] =  \|(\text{d}E_p[A])_p\|_p^2,
    \end{align}
    where the subscript $p^{(\alpha)}$ in Var means variance with respect to $p^{(\alpha)}$.
    $\hfill$% \QEDopen$
  \end{theorem}
    \begin{proof}
    For any tangent vector $X\in T_p(\mathcal{P})$,
    \begin{align}
    \label{eqn:tangent_acting_on_expectation}
    X(E_p[A])
    & = \sum\limits_x X(x)A(x)\nonumber\\
    & = E_p[X_p^{(e)} \cdot A]\\
    & = E_p[X_p^{(e)}(A-E_p[A])].
    \end{align}
    Since $A-E_p[A]\in T_p^{(\alpha)}(\mathcal{P})$ (c.f.~(\ref{e_space_equalto_alpha_space})), there exists $Y\in T_p(\mathcal{P})$ such that $A-E_p[A] = Y_p^{(\alpha)}$, and $\text{grad}(E[A]) = Y$. Hence we see that
    \begin{align*}
    \|(\text{d}E[A])_p\|_p^2 & = E_p[Y_p^{(e)}Y_p^{(\alpha)}]\\
      & = E_p[Y_p^{(e)}(A-E_p[A])]\\
      & \stackrel{(a)}{=} \displaystyle E_p\left[\left\{\frac{p(X)}{ p^{(\alpha)}(X)} (A-E_p[A]) + E_{p^{(\alpha)}}[Y_p^{(e)}]\right\}(A-E_p[A])\right]\\
      & \stackrel{(b)}{=}  E_p\left[\frac{p(X)}{p^{(\alpha)}(X)}(A-E_p[A])(A-E_p[A])\right]\\
      & = E_{p^{(\alpha)}}\left[\frac{p(X)}{p^{(\alpha)}(X)}(A-E_p[A])\frac{p(X)}{p^{(\alpha)}(X)}(A-E_p[A])\right]\\
      & = \text{Var}_{p^{(\alpha)}}\left[\frac{p(X)}{p^{(\alpha)}(X)}(A-E_p[A])\right],
    \end{align*}
    where the equality (a) is obtained by applying (\ref{eqn:alpha_rep_tgt_vec}) to $Y$ and (b) follows because $E_p[A-E_p[A]] = 0$.
    %\begin{flushright}$\blacksquare$\end{flushright}
  \end{proof}

  \begin{corollary}[\cite{kumar2020cram}]
    \label{cor:variance_and_norm_of_differential_inequality}
    If $S$ is a submanifold of $\mathcal{P}$, then
    \begin{align}
    \label{variance_and_norm_of_differential}
    \text{Var}_{p^{(\alpha)}}\left[\frac{p(X)}{p^{(\alpha)}(X)}(A-E_p[A])\right] \ge \|(\text{d}E[A]|_{S})_p\|_p^2
    \end{align}
    with equality if and only if $$A-E_p[A]\in \{X_p^{(\alpha)} : X\in T_p(S)\} =: T_p^{(\alpha)}(S).$$ 
  \end{corollary}

  We use the aforementioned ideas to establish an $\alpha$-version of the CR inequality for the $\alpha$-escort of the underlying distribution. This gives a lower bound for the variance of the unbiased estimator $\hat{\theta}^{(\alpha)}$ in $S^{(\alpha)}$.
  
  \begin{theorem} [$\alpha$-version of Cram\'{e}r-Rao inequality \cite{kumar2020cram}]
  \label{thm:alpha_CRLB}
  Let $S = \{p_{\theta} : \theta = (\theta_1,\dots,\theta_m)\in\Theta\}$ be the given statistical model. Let $\hat{\theta}^{(\alpha)} = (\hat{\theta}^{(\alpha)}_1,\dots,\hat{\theta}^{(\alpha)}_m)$ be an unbiased estimator of $\theta = (\theta_1,\dots,\theta_m)$ for the statistical model $S^{(\alpha)} := \{p_\theta^{(\alpha)} : p_\theta\in S\}$. Then, $\text{Var}_{\theta^{(\alpha)}}[\hat{\theta}^{(\alpha)}(X)] \ge [G^{(\alpha)}]^{-1}$, 
  %\begin{align}
  %\label{eq:analogous_cramer_rao_inequality}
  %\text{Var}_{\theta^{(\alpha)}}[\hat{\theta}^{(\alpha)}(X)] \ge [G^{(\alpha)}]^{-1},
  %\end{align}
  where %$c = (c_1,\dots,c_m)\in \mathbb{R}^m$ and 
  $\theta^{(\alpha)}$ denotes expectation with respect to $p_{\theta}^{(\alpha)}$. On the other hand, given an unbiased estimator $\hat{\theta} = (\hat{\theta}_1,\dots,\hat{\theta}_m)$ of $\theta$ for $S$, there exists an unbiased estimator $\hat{\theta}^{(\alpha)} = (\hat{\theta}^{(\alpha)}_1,\dots,\hat{\theta}^{(\alpha)}_m)$ of $\theta$ for $S^{(\alpha)}$ such that $\text{Var}_{\theta^{(\alpha)}}[\hat{\theta}^{(\alpha)}(X)] \ge [G^{(\alpha)}]^{-1}$.
  
  \noindent(We follow the convention that, for two matrices $M$ and $N$, $M \ge N$ implies that $M-N$ is positive semi-definite.)
  \end{theorem}
  
    \begin{proof}
    Given an unbiased estimator $\hat{\theta}^{(\alpha)} = (\hat{\theta}^{(\alpha)}_1,\dots,\hat{\theta}^{(\alpha)}_m)$ of $\theta = (\theta_1,\dots,\theta_m)$ for the statistical model $S^{(\alpha)}$, let
  \begin{align}
  \label{unbiased_estimator_for_S}
  \hat{\theta_i}(X) := \frac{p_{\theta}^{(\alpha)}(X)}{p_{\theta}(X)}\hat{\theta}_i^{(\alpha)}(X).
  \end{align}
  It is easy to check that $\hat{\theta}$ is an unbiased estimator of $\theta$ for $S$. Hence, if we let $A = \sum\limits_{i=1}^m c_i \hat{\theta_i}$, for $c = (c_1,\dots,c_m)\in \mathbb{R}^m$, then from (\ref{variance_and_norm_of_differential}) and (\ref{differential_and_metric_alpha}), we have
  \begin{align}
  \label{eq:analogous_cramer_rao_inequality}
  c\text{Var}_{\theta^{(\alpha)}}[\hat{\theta}^{(\alpha)}(X)]c^t \ge c[G^{(\alpha)}]^{-1}c^t.
  \end{align}
  This proves the first part.
  
  % On the other hand, given
  For the converse, consider an unbiased estimator $\hat{\theta} = (\hat{\theta}_1,\dots,\hat{\theta}_m)$ of $\theta$ for $S$. Let
  \begin{align}
  \label{unbiased_estimator_alpha_version}
  \hat{\theta}_i^{(\alpha)}(X) := \frac{p_{\theta}(X)}{p_{\theta}^{(\alpha)}(X)}\hat{\theta_i}(X).
  \end{align}
  This is an unbiased estimator of $\theta_i$ for $S^{(\alpha)}$.  Hence, the assertion follows from the first part of the proof. %\begin{flushright}$\blacksquare$\end{flushright}
  \end{proof}
  
  When $\alpha = 1$, the inequality in (\ref{eq:analogous_cramer_rao_inequality}) reduces to the classical Cram\'{e}r-Rao inequality. We see that (\ref{eq:analogous_cramer_rao_inequality}) is, in fact, the Cram\'{e}r-Rao inequality for the $\alpha$-escort family $S^{(\alpha)}$.

\subsection{Generalized version of Cram\'{e}r-Rao inequality}
\label{subsec:general_framework}
We apply the result in (\ref{eq:analogous_cramer_rao_inequality}) to a more general class of $f$-divergences. Observe from (\ref{eqn:the-function-in_D_f}) that $I_\alpha$-divergence is a monotone function of an $f$-divergence not of the actual distributions but their $\alpha$-escort distributions.
Motivated by this, we first define a more general $f$-divergence and then show that these diveregnces also lead to generalized CR inequality analogous to (\ref{eq:analogous_cramer_rao_inequality}). Although these divergences are defined for positive measures, we restrict to probability measures here.
  
  \begin{definition}
  \label{defn:gen_f-divergence}
  Let $f$ be a strictly convex, twice continuously differentiable real valued function defined on $[0,\infty)$ with $f(1) = 0$ and $f''(1)\neq 0$. Let $F$ be a function that maps a probability distribution $p$ to another probability distribution $F(p)$. Then the {\em generalized $f$-divergence} between two probability distributions $p$ and $q$ is defined by
  \begin{equation}
      \label{eqn:gen_f-divergence}
      D_f^{(F)}(p,q) = \frac{1}{f''(1)}\cdot\sum_x F(q(x))f\left(\frac{F(p(x))}{F(q(x))}\right).
  \end{equation}
  \end{definition}
  
  Since $f$ is convex, by Jensen's inequality,
  \begin{align*}
   D_f^{(F)}(p,q) & \ge \frac{1}{f''(1)} \cdot f\left(\sum_x F(q(x))\cdot \frac{F(p(x))}{F(q(x))}\right)\\
   & = \frac{1}{f''(1)} \cdot f\left(\sum_x F(p(x))\right)\\
   & = \frac{1}{f''(1)}\cdot f(1)\\
   & = 0.
  \end{align*}
  Notice that, when $F(p(x)) = p(x)$, $D_f^{(F)}$ becomes the usual Csisz\'ar divergence. We now apply Eguchi's theory to $D_f^{(F)}$. The Riemannian metric on $S$ is specified by the matrix $G^{(f,F)}(\theta) = [g_{i,j}^{(f,F)}(\theta)]$, where\par\noindent\small
  \begin{align}
  %\label{eqn:gen_metric}
  \lefteqn{g_{i,j}^{(f,F)}(\theta) ~ := ~ g_{i,j}^{(D_f^{(F)})}(\theta)} \nonumber\\
  & = -\frac{\partial}{\partial\theta_j'}\frac{\partial}{\partial\theta_i}D_f^{(F)}(p_{\theta},p_{\theta'})\bigg|_{\theta' = \theta} \nonumber\\
    & = - \frac{\partial}{\partial\theta_j'}\frac{\partial}{\partial\theta_i} \sum_x F(p_{\theta'}(x))f\left(\frac{F(p_{\theta}(x))}{F(p_{\theta'}(x))}\right)\bigg|_{\theta' = \theta}\cdot \frac{1}{f''(1)}\nonumber\\
  & = - \frac{\partial}{\partial\theta_j'}\left[\sum_x F(p_{\theta'}(x))f'\left(\frac{F(p_{\theta}(x))}{F(p_{\theta'}(x))}\right)\frac{F'(p_{\theta}(x))}{F(p_{\theta'}(x))}\partial_i p_{\theta}(x)\right]_{\theta' = \theta}\cdot \frac{1}{f''(1)}\nonumber\\
  & = \left[\sum_x F'(p_{\theta}(x)) f''\left(\frac{F(p_{\theta}(x))}{F(p_{\theta'}(x))}\right)\frac{F(p_{\theta}(x))}{F(p_{\theta'}(x))^2}F'(p_{\theta'}(x))\partial_i p_{\theta}(x)\partial_j p_{\theta}(x)\right]_{\theta' = \theta}\nonumber\\
  &\hspace{4mm}\cdot \frac{1}{f''(1)}\nonumber\\
  & = \sum_x F(p_{\theta}(x))\cdot \partial_i \log F(p_{\theta}(x))\cdot \partial_j \log F(p _{\theta}(x))\nonumber\\
    \label{eqn:gen_metric}
  & = E_{\theta^{(F)}}[\partial_i \log F(p_{\theta}(X))\cdot \partial_j \log F(p _{\theta}(X))],
  \end{align}\normalsize
   where $\theta^{(F)}$ stands for expectation with respect to the escort measure $F(p_\theta)$.
   
  Although the \textit{generalized} Csisz\'ar $f$-divergence is also a  Csisz\'ar$f$-divergence, it is not between $p$ and $q$. Rather, it is between the distributions $F(p)$ and $F(q)$. As a consequence, the metric induced by $D_f^{(F)}$ is different from the Fisher information metric, whereas the metric arising from all Csisz\'ar $f$-divergences is the Fisher information metric \cite{amari2010information}.
  
  The following theorem extends the result in Theorem \ref{thm:alpha_CRLB} to a more general framework.
  
  \begin{theorem}[Generalized version of Cram\'{e}r-Rao inequality \cite{kumar2020cram}]
  \label{thm:gen_crlb}
  Let $\hat{\theta} = (\hat{\theta}_1,\dots,\hat{\theta}_m)$ be an unbiased estimator of $\theta = (\theta_1,\dots,\theta_m)$ for the statistical model $S$. Then there exists an unbiased estimator $\hat{\theta}^{F}$ of $\theta$ for the model $S^{(F)} = \{F(p) : p\in S\}$ such that $\text{Var}_{\theta^{(F)}}[\hat{\theta}^{(F)}(X)] \ge [G^{(f,F)}]^{-1}$. Further, if $S$ is such that its escort model $S^{(F)}$ is exponential, then there exists efficient estimators for the escort model.
  \end{theorem}
    \begin{proof}
  Following the same steps as in Theorems \ref{thm:variance_and_norm_of_differential}-\ref{thm:alpha_CRLB} and Corollary \ref{cor:variance_and_norm_of_differential_inequality} produces 
  \begin{align}
  \label{eq:analogous_cramer_rao_inequality2}
  c\text{Var}_{\theta^{(F)}}[\widehat{\theta}^{(F)}]c^t \ge c[G^{(f,F)}]^{-1}c^t
  \end{align}
   for an unbiased estimator $\widehat{\theta}^{(F)}$ of $\theta$ for $S^{(F)}$. This proves the first assertion of the theorem. Now let us suppose that $p_\theta$ is model such that
  \begin{equation}
  \label{eqn:escort_model}
      \log F(p_\theta(x)) = c(x) + \sum_{i=1}^k \theta_i h_i(x) - \psi(\theta).
  \end{equation}
  % That is, $p_\theta$ is any model 
  Then
  \begin{equation}
  \label{eqn:derivative_escort_model}
      \partial_i \log F(p _{\theta}(x)) = h_i(x) - \psi(\theta).
  \end{equation}
  Let
    \[  
  \widehat{\eta}(x) := h_i(x) \quad \text{and} \quad \eta := E_{\theta^{(F)}}[\widehat{\eta}(X)].
  \]
  Since $E_{\theta^{(F)}}[\partial_i \log F(p _{\theta}(X))] = 0$, we have
  \begin{equation}
      \label{eqn:derivative_of_potential}
      \partial_i \psi(\theta) = \eta_i.
  \end{equation}
  Hence
  \begin{equation}
  \label{eqn:gen_metric_dual_expression1}
      g_{i,j}^{(f,F)}(\theta) = E_{\theta^{(F)}}[(\widehat{\eta}_i(X) - \eta_i)(\widehat{\eta}_j(X) - \eta_j)]
  \end{equation}
  Moreover, since
  \begin{align}
  \label{eqn:second_derivative}
  \partial_i \partial_j\log F(p _{\theta}(x)) & = \partial_i\left[\frac{1}{F(p _{\theta}(x))}\partial_j \log F(p _{\theta}(x))\right]\nonumber\\
  & = \frac{1}{F(p _{\theta}(x))}\partial_i \partial_j F(p _{\theta}(x)) - \frac{1}{F(p _{\theta}(x))^2}\partial_i F(p _{\theta}(x))\partial_j F(p _{\theta}(x))\nonumber\\
    & = \frac{1}{F(p _{\theta}(x))}\partial_i \partial_j F(p _{\theta}(x)) - \partial_i\log F(p _{\theta}(x))\partial_j\log F(p _{\theta}(x)),
  \end{align}
from (\ref{eqn:gen_metric}), we have
\begin{equation}
  \label{eqn:gen_metric_dual_expression2}
      g_{i,j}^{(f,F)}(\theta) = - E_{\theta^{(F)}}[\partial_i \partial_j\log F(p _{\theta}(X))].
  \end{equation}  
  Hence, from (\ref{eqn:derivative_escort_model}) and (\ref{eqn:derivative_of_potential}), we have
  \begin{equation}
      \label{eqn:dual_equation}
      \partial_i \eta_j = g_{i,j}^{(f,F)}(\theta).
  \end{equation}
  This implies that $\eta$ is dual to $\theta$. Hence the generalized FIM of $\eta$ is equal to the inverse of the generalized FIM of $\theta$. Thus from
  (\ref{eqn:gen_metric_dual_expression1}), $\widehat{\eta}$ is an efficient estimator of $\eta$ for the escort model. This further helps us to find efficient estimators for $\theta$ for the escort model. This completes the proof. 
  \end{proof}
    
  Theorem~\ref{thm:gen_crlb} generalizes the dually flat structure of exponential and linear families with respect to the Fisher metric identified by Amari and Nagoaka
  \cite[Sec.~3.5]{amari2000methods} to other distributions and a more widely applicable metric (as in Definition~\ref{defn:gen_f-divergence}).

\section{Information Geometry for Bayesian CR inequality and Barankin Bound}
\label{sec:bcrlb}
We extend Eguchi's theory in Section~\ref{sec:desiderata} to the space $\tilde{\mathcal{P}}(\mathbb{X})$ of all positive measures on $\mathbb{X}$, that is, $\tilde{\mathcal{P}} = \{\tilde{p}:\mathbb{X}\to (0,\infty)\}$. Let $S = \{p_\theta: \theta = (\theta_1,\dots,\theta_k)\in\Theta\}$ be a $k$-dimensional sub-manifold of $\mathcal{P}$ and let
\begin{align}
\label{eqn:denormalized_manifold}
\tilde{S} := \{\tilde{p}_{\theta}(x) = p_{\theta}(x)\lambda(\theta) : p_{\theta}\in S\},
\end{align}
where $\lambda$ is a probability distribution on $\Theta$. Then $\tilde{S}$ is a sub-manifold of $\tilde{\mathcal{P}}$. For $\tilde{p}_\theta, \tilde{p}_{\theta'}\in \tilde{\mathcal{S}}$, the KL-divergence between $\tilde{p}_\theta$ and $\tilde{p}_{\theta'}$ is given by
\begin{align}
  \label{eqn:KL-div}
  I(\tilde{p}_\theta\|\tilde{p}_{\theta'}) & = \sum_x \tilde{p}_{\theta}(x) \log \frac{\tilde{p}_{\theta}(x)}{\tilde{p}_{\theta'}(x)} - \sum_x \tilde{p}_{\theta}(x) + \sum_x \tilde{p}_{\theta'}(x)\nonumber\\
  & = \sum_x p_{\theta}(x)\lambda(\theta) \log \frac{p_{\theta}(x)\lambda(\theta)}{p_{\theta'}(x)\lambda(\theta')} - \lambda(\theta) + \lambda(\theta').\nonumber\\
\end{align}
By following Eguchi, we define a Riemannian metric $G^{(I)}(\theta) = [g_{i,j}^{(I)}(\theta)]$ on $\tilde{S}$ by
\begin{align}
  g_{i,j}^{(I_\lambda)}(\theta)
  & := - I[\partial_i\|\partial_j]\nonumber\\
  & = \left. - \frac{\partial}{\partial \theta_i} \frac{\partial}{\partial \theta'_j} \sum_{x} p_{\theta}(x)\lambda(\theta) \log \frac{p_{\theta}(x)\lambda(\theta)}{p_{\theta'}(x)\lambda(\theta')} \right|_{\theta' = \theta}\nonumber\\
  & = \sum_x \partial_i (p_{\theta}(x)\lambda(\theta))\cdot \partial_j \log (p _{\theta}(x)\lambda(\theta))\nonumber\\
  & = \sum_x p_{\theta}(x)\lambda(\theta)\partial_i (\log p_{\theta}(x)\lambda(\theta))\cdot \partial_j (\log (p _{\theta}(x)\lambda(\theta)))\nonumber\\
  & = \lambda(\theta)\sum_x p_{\theta}(x)[\partial_i (\log p_{\theta}(x)) + \partial_i (\log \lambda(\theta))]\nonumber\\
  & \hspace*{2cm}\cdot [\partial_j (\log p_{\theta}(x)) + \partial_j (\log \lambda(\theta))]\nonumber\\
  & = \lambda(\theta) \big\{E_{\theta}[\partial_i \log p_{\theta}(X)\cdot \partial_j \log p_{\theta}(X)]\nonumber\\
  & \hspace*{2cm}\cdot  + \partial_i (\log \lambda(\theta))\cdot\partial_j (\log \lambda(\theta))\big\}\\
  \label{eqn:fisher-information-metric-bayesian}
  & = \lambda(\theta) \big\{g_{i,j}^{(e)}(\theta) + J_{i,j}^{\lambda}(\theta)\big\},
\end{align}
where
\begin{equation}
\label{eqn:Fisher_metric_bayesian}
g_{i,j}^{(e)}(\theta) := E_{\theta}[\partial_i \log p_{\theta}(X)\cdot \partial_j \log p_{\theta}(X)],
\end{equation}
and
\begin{equation}
\label{eqn:J_matrix_bayesian}
J_{i,j}^{\lambda}(\theta) := \partial_i (\log \lambda(\theta))\cdot\partial_j (\log \lambda(\theta)).
\end{equation} 
Let $G^{(e)}(\theta) := [g^{(e)}_{i,j}(\theta)]$ and $J^{\lambda}(\theta) := [J_{i,j}^{\lambda}(\theta)]$. Then 
\begin{equation}
\label{eqn:Bayesian_Fisher}
G^{(I)}(\theta) = \lambda(\theta)\big[G^{(e)}(\theta) + J^{\lambda}(\theta)\big],
\end{equation}
where $G^{(e)}(\theta)$ is the usual FIM. Observe that $\tilde{\mathcal{P}}$ is an affine subset of $\mathbb{R}^{\tilde{\mathbb{X}}}$, where $\tilde{\mathbb{X}} := \mathbb{X}\cup\{a_{d+1}\}$. The tangent space at every point of $\tilde{\mathcal{P}}$ is $\mathcal{A}_0 := \{A\in\mathbb{R}^{\tilde{\mathbb{X}}} : \sum_{x\in\tilde{\mathbb{X}}}A(x) = 0\}$. That is, $T_p(\tilde{\mathcal{P}}) = \mathcal{A}_0$. Thus, proceeding with Amari and Nagoaka's theory \cite[sec.~2,5]{amari2000methods} (as in subsection \ref{subsec:alphaCRLB}) with $p$ replaced by $\tilde{p}$, we get the following theorem and corollary. 

\begin{thm}[\cite{kumar2018information}]
  Let $A:\mathbb{X}\to\mathbb{R}$ be any mapping (that is, a vector in $\mathbb{R}^{\mathbb{X}}$.) Let $E[A]:\mathcal{\tilde{P}}\to \mathbb{R}$ be the mapping $\tilde{p}\mapsto E_{\tilde{p}}[A]$. We then have
  \begin{align}
    \label{eqn:variance_eq_norm_of_differential}
    \text{Var}(A) =  \|(\text{d}E_{\tilde{p}}[A])_{\tilde{p}}\|_{\tilde{p}}^2.
  \end{align}
\end{thm}

\begin{corollary}[\cite{kumar2018information}]
  \label{cor:c_r_inequality}
  If $S$ is a submanifold of $\mathcal{\tilde{P}}$, then
  \begin{align}
    \label{eqn:variance_ge_norm_of_differential}
    \text{Var}_{\tilde{p}}[A] \ge \|(\text{d}E[A]|_{S})_{\tilde{p}}\|_{\tilde{p}}^2
  \end{align}
  with equality iff $$A-E_{\tilde{p}}[A]\in \{X_{\tilde{p}}^{(e)} : X\in T_{\tilde{p}}(S)\} =: T_{\tilde{p}}^{(e)}(S).$$ 
\end{corollary}

%\subsection{Derivation of Error Bounds}
%\label{sec:error_bounds}
We state our main result in the following theorem.

\begin{theorem}[\cite{kumar2018information}]
\label{thm:main}
  Let $S$ and $\tilde{S}$ be as in (\ref{eqn:denormalized_manifold}). Let $\widehat{\theta}$ be an estimator of $\theta$. Then
  
  \begin{enumerate}
    \item[(a)] {\em Bayesian Cram\'{e}r-Rao}: 
    \begin{equation}
    \label{eqn:Bayesian_Cramer-Rao}
    \mathbb{E}_\lambda\big[\text{Var}_\theta(\widehat{\theta})\big] \ge \big\{\mathbb{E}_\lambda [G^{(e)}(\theta) + J^{\lambda}(\theta)]\big\}^{-1}, 
    \end{equation}
    where $\text{Var}_\theta(\widehat{\theta}) = [\text{Cov}_{\theta}(\widehat{\theta}_i(X),\widehat{\theta}_j(X))]$  is the covariance matrix and $G^{(e)}(\theta)$ and $J^{\lambda}(\theta)$ are as in (\ref{eqn:Fisher_metric_bayesian}) and (\ref{eqn:J_matrix_bayesian}). 
     
    \item[(b)] {\em Deterministic Cram\'{e}r-Rao}: If $\widehat{\theta}$ is an unbiased estimator of $\theta$, then
    \begin{equation}
    \label{eqn:Deter_Cramer_Rao}
    \text{Var}_\theta[\widehat{\theta}] \ge [G^{(e)}(\theta)]^{-1}.
    \end{equation}
    
    \item[(c)] {\em Deterministic Cram\'{e}r-Rao (biased version)}: For any estimator $\widehat{\theta}$ of $\theta$,
    \begin{align*}
    \text{MSE}_\theta[\widehat{\theta}] \ge (\textbf{1} + B'(\theta)) [G^{(e)}(\theta)]^{-1}(\textbf{1} + B'(\theta))\\
    + b(\theta)b(\theta)^T,
    \end{align*}
    where $b(\theta) = (b_1(\theta),\dots,b_k(\theta))^T := \mathbb{E}_\theta[\widehat{\theta}]$ is the bias and $\textbf{1} + B'(\theta)$ is the matrix whose $(i,j)$th entry is $0$ if $i\neq j$ and is $(1+\partial_i b_i(\theta))$ if $i=j$.
    
    \item[(d)] {\em Barankin Bound}: (Scalar case) If $\widehat{\theta}$ be an unbiased estimator of $\theta$, then
    \begin{equation}
    \label{eqn:Barankin_Bound}
    \text{Var}_\theta[\widehat{\theta}] \ge \sup_{n,a_l,\theta^{(l)}}\frac{\Big[\sum\limits_{l=1}^n a_l(\theta^{(l)}-\theta)\Big]^2}{\sum\limits_x \Big[\sum\limits_{l=1}^n a_l L_{\theta^{(l)}}(x)\Big]^2 p_\theta(x)},
    \end{equation}
    where $L_{\theta^{(l)}}(x) := {p_{\theta^{(l)}}(x)}/{p_\theta(x)}$ and the supremum is over all $a_1,\dots,a_n\in\mathbb{R}$, $n\in\mathbb{N}$, and $\theta^{(1)},\dots,\theta^{(n)}\in\Theta$.
  \end{enumerate}
\end{theorem}
\begin{proof}
  \begin{enumerate}
    \item[(a)] Let $A = \sum_{i=1}^{k}c_i\widehat{\theta}_i$, where $\widehat{\theta} = (\widehat{\theta}_1,\dots,\widehat{\theta}_k)$ is an unbiased estimator of $\theta$, in Corollary \ref{cor:c_r_inequality}. Then, from (\ref{eqn:variance_ge_norm_of_differential}), we have
    \begin{equation*}
    \sum\limits_{i,j} c_i c_j \text{Cov}_{\tilde{\theta}}(\widehat{\theta}_i,\widehat{\theta}_j)\ge \sum\limits_{i,j} c_i c_j (g^{(I)})^{i,j}(\theta).
    \end{equation*}
    This implies that
      \begin{equation}
      \label{eqn:cramer_rao_general}
    \lambda(\theta)\sum\limits_{i,j} c_i c_j \text{Cov}_{\theta}(\widehat{\theta}_i,\widehat{\theta}_j)\ge \sum\limits_{i,j} c_i c_j (g^{(I)})^{i,j}(\theta).
    \end{equation}
    Hence, integrating with respect to $\theta$, from (\ref{eqn:Bayesian_Fisher}), we get
    \begin{align*}
    \lefteqn{\sum\limits_{i,j} c_i c_j \mathbb{E}_\lambda\big[\text{Cov}_{\theta}(\widehat{\theta}_i,\widehat{\theta}_j)\big]}\\
        & \ge \sum\limits_{i,j} c_i c_j \mathbb{E}_\lambda\big[[G^{(e)}(\theta) + J^{\lambda}(\theta)]^{-1}\big].
    \end{align*}
    That is,
    \begin{equation*}
    \mathbb{E}_\lambda\big[\text{Var}_\theta(\widehat{\theta})\big] \ge \mathbb{E}_\lambda\big[[G^{(e)}(\theta) + J^{\lambda}(\theta)]^{-1}\big].
    \end{equation*}
    But
        \begin{equation*}
        \mathbb{E}_\lambda\big[[G^{(e)}(\theta) + J^{\lambda}(\theta)]^{-1}\big] \ge \big[\mathbb{E}_\lambda[G^{(e)}(\theta) + J^{\lambda}(\theta)]\big]^{-1}
        \end{equation*}
        by \cite{GrovesRothenberg1969Biometrika}. This proves the result.
    
    \item[(b)] This follows from (\ref{eqn:cramer_rao_general}) by taking $\lambda(\theta) = 1$.
    
        \item[(c)] Let us first observe that $\widehat{\theta}$ is an unbiased estimator of $\theta + b(\theta)$. Let $A = \sum_{i=1}^{k} c_i\widehat{\theta}_i$ as before. Then $\mathbb{E}[A] = \sum_{i=1}^{k}c_i(\theta_i + b_i(\theta))$. Then, from Corollary \ref{cor:c_r_inequality} and (\ref{eqn:variance_ge_norm_of_differential}), we have
    \begin{equation*}
     \text{Var}_\theta[\widehat{\theta}] \ge (\textbf{1} + B'(\theta)) [G^{(e)}(\theta)]^{-1}(\textbf{1} + B'(\theta))
    \end{equation*}
    But $\text{MSE}_\theta[\widehat{\theta}] = \text{Var}_\theta[\widehat{\theta}] + b(\theta)b(\theta)^T$. This proves the assertion.
        
    \item[(d)] For fixed $a_1,\dots,a_n\in\mathbb{R}$ and $\theta^{(1)},\dots,\theta^{(n)}\in\Theta$, let us define a metric by the following formula
    \begin{equation}
    \label{Eqn:Metric_Barankin}
    g(\theta) := \sum\limits_x \Big[\sum\limits_{l=1}^n a_l L_{\theta^{(l)}}(x)\Big]^2 p_\theta(x).
    \end{equation}
    Let $f$ be the mapping $p\mapsto\mathbb{E}_p[A]$. Let $A(\cdot) = \widehat{\theta}(\cdot) - \theta$, where $\widehat{\theta}$ is an unbiased estimator of $\theta$ in Corollary \ref{cor:c_r_inequality}. Then, from (\ref{eqn:variance_ge_norm_of_differential}), we have
    \begin{align*}
    \sum\limits_x(\widehat{\theta}(x) - \theta)\Big(\sum\limits_{l=1}^na_lL_{\theta^{(l)}}(x)\Big)p_\theta(x)\\
    &\hspace*{-4cm} = \sum\limits_{l=1}^n a_l \Big(\sum\limits_x(\widehat{\theta}(x) - \theta)\frac{p_{\theta^{(l)}}(x)}{p_\theta(x)}p_\theta(x)\Big)\\
    &\hspace*{-4cm} = \sum\limits_{l=1}^n a_l (\theta^{(l)} - \theta).
    \end{align*}
    Hence, from Corollary \ref{cor:c_r_inequality}, we have
     \begin{equation*}
     \text{Var}_\theta[\widehat{\theta}] \ge \frac{\Big[\sum\limits_{l=1}^n a_l(\theta^{(l)}-\theta)\Big]^2}{\sum\limits_x \Big[\sum\limits_{l=1}^n a_l L_{\theta^{(l)}}(x)\Big]^2 p_\theta(x)}
     \end{equation*}
     Since $a_l$ and $\theta^{(l)}$ are arbitrary, taking supremum over all $a_1,\dots,a_n\in\mathbb{R}$, $n\in\mathbb{N}$, and $\theta^{(1)},\dots,\theta^{(n)}\in\Theta$, we get (\ref{eqn:Barankin_Bound}).
  \end{enumerate}
\end{proof} 

\section{Information Geometry For Bayesian \texorpdfstring{$\alpha$}{}-CR inequality}
\label{sec:relalphaBayes}
We now introduce $I_\alpha$-divergence in the Bayesian case. Consider the setting of Section~\ref{sec:bcrlb}. Then, $I_\alpha$-divergence between $\tilde{p}_\theta$ with respect to $\tilde{p}_{\theta'}$ is
\begin{align}
\label{eq:relent_alpha_bayesian}
 I_{\alpha}(\tilde{p}_{\theta},\tilde{p}_{\theta'})
    & := \frac{\lambda(\theta)}{1-\alpha}\log\sum_x p_\theta(x) (\lambda(\theta')p_{\theta'}(x))^{\alpha-1}+ \lambda(\theta')\nonumber\\
&  - \lambda(\theta)\left[\frac{\log\sum_xp_\theta(x)^\alpha}{\alpha (1-\alpha)} - \{1 + \log \lambda(\theta)\} -\frac{1}{\alpha} \log\sum_xp_{\theta'}(x)^\alpha\right].\nonumber
\end{align}

We present the following Lemma~\ref{lem:rel} which shows that our definition of Bayesian $I_\alpha$-divergence is not only a valid divergence function but also coincides with the KL-divergence as $\alpha\to 1$.
\begin{lemma}[\cite{mishra2020generalized}]
\label{lem:rel}
\begin{enumerate}
\item[]
\item $I_{\alpha}(\tilde{p}_{\theta},\tilde{p}_{\theta'})\ge 0$ with equality if and only if $\tilde{p}_{\theta} = \tilde{p}_{\theta'}$
\item $I_{\alpha}(\tilde{p}_{\theta},\tilde{p}_{\theta'})\to I(\tilde{p}_\theta,\tilde{p}_{\theta'})$ as $\alpha\to 1$.
\end{enumerate}

\end{lemma}
\begin{proof}
1) Let $\alpha >1$. Applying Holder's inequality with Holder conjugates $p=\alpha$ and $q={\alpha}/{(\alpha-1)}$, we have
    \[
    \sum_x p_\theta(x) (\lambda(\theta')p_{\theta'}(x))^{\alpha-1} \le \|p_\theta\| \lambda(\theta')^{\alpha-1}\|p_\theta'\|^{\alpha-1},
    \]
    where $\|\cdot\|$ denotes $\alpha$-norm. When $\alpha <1$, the inequality is reversed. Hence
    \begin{flalign*}
        &\lefteqn{\frac{\lambda(\theta)}{1-\alpha}\log\sum_x p_\theta(x) (\lambda(\theta')p_{\theta'}(x))^{s-1}}\nonumber\\
        &\ge \frac{\lambda(\theta)\log\sum_xp_\theta(x)^\alpha}{\alpha(1-\alpha)} - \lambda(\theta)\log \lambda(\theta') - \frac{\lambda(\theta)}{\alpha}\log\sum_xp_{\theta'}(x)^\alpha\nonumber\\
& \ge \frac{\lambda(\theta)\log\sum_xp_\theta(x)^\alpha}{\alpha(1-\alpha)} - \lambda(\theta)\log \lambda(\theta) - \lambda(\theta) + \lambda(\theta')\nonumber\\
& \hfill - \frac{\lambda(\theta)}{\alpha}\log\sum_xp_{\theta'}(x)^\alpha\\
&= \lambda(\theta)\left[\frac{\log\sum_xp_\theta(x)^\alpha}{\alpha(1-\alpha)} - \{1 + \log \lambda(\theta)\} - \log\sum_xp_{\theta'}(x)^\alpha\right] + \lambda(\theta'),\nonumber
\end{flalign*}
where the second inequality follows because, for $x,y\ge 0$,
\begin{align*}
\log\frac{x}{y} = -x\log\frac{y}{x} \ge -x (y/x-1) \ge -y + x,
\end{align*}
and hence $$x\log y \le x\log x - x + y.$$
The conditions of equality follow from the same in Holder's inequality and $\log x\le x-1$.

  2) This follows by applying L'H\^{o}pital rule to the first term of $I_\alpha$:
\begin{flalign}
    &\lim_{\alpha\to 1}\left[\frac{\alpha}{1-\alpha} \lambda(\theta) \log \sum_x p_{\theta}(x) (\lambda(\theta')p_{\theta'}(x))^{\alpha-1}\right]\nonumber\\
    &=\lim_{\alpha\to 1}\left[\frac{1}{\frac{1}{\alpha}-1} \lambda(\theta) \log \sum_x p_{\theta}(x) (\lambda(\theta')p_{\theta'}(x))^{\alpha-1}\right]\nonumber\\
    &=\lambda(\theta)\lim_{\alpha\to 1}\left[\frac{1}{-\frac{1}{\alpha^2}}\frac{\sum_x p_{\theta}(x) (\lambda(\theta')p_{\theta'}(x))^{\alpha-1} \log(\lambda(\theta')p_{\theta'}(x))}{\sum_x p_{\theta}(x) (\lambda(\theta')p_{\theta'}(x))^{\alpha-1}}\right]\nonumber\\
    &= -\sum_x (\lambda(\theta)p_{\theta}(x)) \log(\lambda(\theta')p_{\theta'}(x)),\nonumber
\end{flalign}
and since R\'{e}nyi entropy coincides with Shannon entropy as $\alpha\to 1$.
\end{proof}

We apply Eguchi's theory provided in Section~\ref{sec:desiderata} to the space $\tilde{\mathcal{P}}(\mathbb{X})$ of all positive measures on $\mathbb{X}$, that is, $\tilde{\mathcal{P}} = \{\tilde{p}:\mathbb{X}\to (0,\infty)\}$. 
Following Eguchi \cite{eguchi1992geometry}, we define a Riemannian metric
$[g_{i,j}^{(I_{\alpha})}(\theta)]$ on $\tilde{S}$ by
%\begin{widetext}
\begin{flalign}
  &{g_{i,j}^{(I_{\alpha})}(\theta)}\nonumber\\
  &=-\frac{\partial}{\partial\theta_j'}\frac{\partial}{\partial\theta_i}I_{\alpha}(\tilde{p}_{\theta},\tilde{p}_{\theta'})\bigg|_{\theta' = \theta}\nonumber\\
  &=  \frac{1}{\alpha-1}\cdot\partial_j'\partial_i \lambda(\theta)\log \sum_y p_{\theta}(x) ({\lambda(\theta') p_{\theta'}(x)})^{\alpha-1}\bigg|_{\theta' = \theta}\nonumber\\
  &\hspace{4mm} - \partial_i\lambda(\theta)\partial_j'\log\sum_x p_{\theta'}(x)^\alpha\bigg|_{\theta' = \theta}\nonumber\\
  \label{eqn:alpha-metric-bayesian}
    & = \frac{1}{\alpha-1}\left\{\lambda(\theta)\sum_x \partial_i p_{\theta}(x)\cdot \partial_j'\left[\frac{({\lambda(\theta') p_{\theta'}(x)})^{\alpha-1}}{\sum_y p_{\theta}(y) ({\lambda(\theta') p_{\theta'}(x)})^{\alpha-1}}\right]_{\theta' = \theta}\right. \nonumber\\
    & \hspace{15mm} \left. + \partial_i\lambda(\theta)\cdot \left[\frac{{\sum_x p_\theta(x) \partial_j'\big(\lambda(\theta)p_{\theta'}(x)}\big)^{\alpha-1}}{\sum_x p_{\theta}(x) ({\lambda(\theta') p_{\theta'}(x)})^{\alpha-1}}\right]_{\theta' = \theta}\right\} \nonumber\\
    &\hspace{4mm} - \partial_i\lambda(\theta)\partial_j'\log\sum_xp_{\theta'}(x)^\alpha\bigg|_{\theta' = \theta}\\
  & = \lambda(\theta) \left\{ \frac{\sum_x\partial_i p_{\theta}(x){(\lambda(\theta)p_{\theta}(x)})^{\alpha-2}\partial_j (\lambda(\theta)p_{\theta}(x))}{\sum_x p_{\theta}(x)(\lambda(\theta)p_{\theta}(x))^{\alpha-1}} \right.\nonumber\\
  & \hspace{14mm} - \frac{\sum_x (\partial_i p_{\theta}(x)){p_{\theta}(x)}^{\alpha-1}}{\sum_x p_{\theta}(x)^{\alpha}}
  \cdot \frac{\sum_x p_\theta(x) (\lambda(\theta)p_{\theta'}(x))^{\alpha-2}\partial_j(\lambda(\theta)p_{\theta}(x))}{\sum_x p_{\theta}(x) ({\lambda(\theta) p_{\theta}(x)})^{\alpha-1}} \nonumber\\
  &\hspace{14mm}\left.+ \partial_i\log\lambda(\theta)\cdot \left[\frac{{\sum_x p_\theta(x) \partial_j'\big(\lambda(\theta)p_{\theta'}(x)}\big)^{\alpha-1}}{\sum_x p_{\theta}(x) ({\lambda(\theta') p_{\theta'}(x)})^{\alpha-1}}\right]_{\theta' = \theta}\right\} \nonumber\\
  &\hspace{4mm}- \partial_i\lambda(\theta)E_{\theta^{(\alpha)}}[\partial_j \log p_{\theta}(X)]\nonumber\\
  & = \lambda(\theta) \left\{E_{\theta^{(\alpha)}}[\partial_i \log p_{\theta}(X)\partial_j \log p_{\theta}(X)] + \partial_j\log\lambda(\theta)E_{\theta^{(\alpha)}}[\partial_i \log p_{\theta}(X)] \right.\nonumber\\
  &\hspace{12mm} - E_{\theta^{(\alpha)}}[\partial_i \log p_{\theta}(X)]\left[E_{\theta^{(\alpha)}}[\partial_j \log p_{\theta}(X)]+ \partial_j\log\lambda(\theta)\right]\nonumber\\
  &\hspace{12mm}\left. + \partial_i\log\lambda(\theta)\cdot \left[E_{\theta^{(\alpha)}}[\partial_j \log p_{\theta}(X)] +  {\partial_j\log\lambda(\theta)}\right]\right\} \nonumber\\
  &\hspace{4mm}- \partial_i\lambda(\theta)E_{\theta^{(\alpha)}}[\partial_j \log p_{\theta}(X)]\nonumber\\
  & = \lambda(\theta) \left[\text{Cov}_{\theta^{(\alpha)}}[\partial_i \log p_{\theta}(X), \partial_j \log p_{\theta}(X)]\right. \nonumber\\
  & \hspace{10mm}\left.+ \partial_i\log\lambda(\theta)\cdot \{E_{\theta^{(\alpha)}}[\partial_j \log p_{\theta}(X)] + \partial_j\log\lambda(\theta)\}\right] \nonumber\\
  &\hspace{4mm}- \partial_i\lambda(\theta)E_{\theta^{(\alpha)}}[\partial_j \log p_{\theta}(X)]\nonumber\\
  & = \lambda(\theta) \left\{\text{Cov}_{\theta^{(\alpha)}}[\partial_i \log p_{\theta}(X), \partial_j \log p_{\theta}(X)] + \partial_i\log\lambda(\theta)\partial_j\log\lambda(\theta)\right\}\nonumber\\
  & = \lambda(\theta)[g_{i,j}^{(\alpha)}(\theta) + J_{i,j}^{\lambda}(\theta)],
    \end{flalign}
%\end{widetext}
where
\begin{equation}
\label{eqn:Fisher_metric_bayesianalpha}
g_{i,j}^{(\alpha)}(\theta) := \text{Cov}_{\theta^{(\alpha)}}[\partial_i \log p_{\theta}(X), \partial_j \log p_{\theta}(X)],
\end{equation}
and
\begin{equation}
\label{eqn:J_matrix_bayesianalpha}
J_{i,j}^{\lambda}(\theta) := \partial_i (\log \lambda(\theta))\cdot\partial_j (\log \lambda(\theta)).
\end{equation} 
Let $G^{(\alpha)}(\theta) := [g^{(\alpha)}_{i,j}(\theta)]$, $J^{\lambda}(\theta) := [J_{i,j}^{\lambda}(\theta)]$ and 
$G_\alpha^{\lambda}(\theta) := G^{(\alpha)}(\theta) + J^{\lambda}(\theta)$.
Notice that, when $\alpha = 1$, $G_\alpha^{\lambda}$ becomes $G^{(I)}$, the usual FIM in the Bayesian case [c.f. \cite{kumar2018information}].

%\subsection{An \texorpdfstring{$\alpha$}{}-Version of Cram\'{e}r-Rao Inequality in the Bayesian Setting}
% \label{sec:analogous_cr_inequality}
  Examining the geometry of $\tilde{\mathcal{P}}$ with respect to the metric $G_\alpha^{\lambda}$, we have the following results analogous to Theorem~\ref{thm:variance_and_norm_of_differential} and Corollary~\ref{cor:variance_and_norm_of_differential_inequality} derived in Section~\ref{subsec:alphaCRLB} for $\mathcal{P}$.
  
  \begin{theorem}\cite{mishra2020generalized}
  \label{thm:variance_and_norm_of_differential_alphaBayesian}
    Let $A:\mathbb{X}\to\mathbb{R}$ be any mapping (that is, a vector in $\mathbb{R}^{\mathbb{X}}$. Let $E[A]:\mathcal{\tilde{P}}\to \mathbb{R}$ be the mapping $\tilde{p}\mapsto E_{\tilde{p}}[A]$. We then have
    \begin{align}
    \label{eqn:variance_and_norm_of_differential_alphaBayesian}
    \text{Var}_{p^{(\alpha)}}\left[\frac{\tilde{p}}{p^{(\alpha)}}(A-E_{\tilde{p}}[A])\right] =  \|(\text{d}E_{\tilde{p}}[A])_{\tilde{p}}\|_{\tilde{p}}^2.
    \end{align}
    $\hfill$% \QEDopen$
  \end{theorem}

  \begin{corollary}\cite{mishra2020generalized}
    \label{cor:variance_and_norm_of_differential_inequality_alphaBayesian}
    If $\tilde{S}$ is a submanifold of $\mathcal{\tilde{P}}$, then
    \begin{align}
    \label{variance_and_norm_of_differential_alphaBayesian}
    \text{Var}_{{p}^{(\alpha)}}\left[\frac{\tilde{p}(X)}{p^{(\alpha)}(X)}(A-E_{\tilde{p}}[A])\right] \ge \|(\text{d}E[A]|_{S})_{\tilde{p}}\|_{\tilde{p}}^2
    \end{align}
    with equality if and only if $$A-E_{\tilde{p}}[A]\in \{X_{\tilde{p}}^{(\alpha)} : X\in T_{\tilde{p}}(S)\} =: T_{\tilde{p}}^{(\alpha)}(S).$$ 
  \end{corollary}
  
  We use the aforementioned ideas to establish a Bayesian $\alpha$-version of the CR inequality for the $\alpha$-escort of the underlying distribution. The following theorem gives a Bayesian lower bound for the variance of an estimator of $S^{(\alpha)}$ starting from an unbiased estimator of $S$.

\begin{theorem} [Bayesian $\alpha$-Cram\'{e}r-Rao inequality \cite{mishra2020generalized}]
  \label{thm:Bayesian_alpha_CRLB}
  Let $S = \{p_{\theta} : \theta = (\theta_1,\dots,\theta_m)\in\Theta\}$ be the given statistical model and let $\tilde{S}$ be as before. Let $\hat{\theta} = (\hat{\theta}_1,\dots,\hat{\theta}_m)$ be an unbiased estimator of $\theta = (\theta_1,\dots,\theta_m)$ for the statistical model $S$. Then
  \begin{equation}
  \label{eqn:Bayesian_alpha_cramerrao}
     \int \text{Var}_{\theta^{(\alpha)}}\left[\frac{\tilde{p_\theta}(X)}{p_\theta^{(\alpha)}(X)}(\hat{\theta}(X) - \theta)\right] d\theta
      \ge \left\{E_\lambda\big[G_\lambda^{(\alpha)}\big]\right\}^{-1},
  \end{equation}
  where $\theta^{(\alpha)}$ denotes expectation with respect to $p_{\theta}^{(\alpha)}$. %(In (\ref{eqn:Bayesian_alpha_cramerrao}), we use the usual convention that, for two matrices $A$ and $B$, $A\ge B$ means that $A-B$ is positive semi-definite.)
  \end{theorem}
\begin{proof}
Given an unbiased estimator $\hat{\theta}$ of $\theta$ for $\tilde{S}$, let $A = \sum\limits_{i=1}^m c_i \hat{\theta_i}$, for $c = (c_1,\dots,c_m)\in \mathbb{R}^m$.

Then, from (\ref{variance_and_norm_of_differential}) and
    \begin{align}
  \label{differential_and_metric_bayesianalpha}
  \|(\text{d}f)_{\tilde{p}}\|_{\tilde{p}}^2 = \sum\limits_{i,j} (g^{i,j})^{(\alpha)}\partial_j(f) \partial_i(f),
  \end{align}
we have
  \begin{align}
  \label{eq:analogous_cramer_rao_inequality1}
  c\text{Var}_{\theta^{(\alpha)}}\left[\frac{\tilde{p_\theta}(X)}{p_\theta^{(\alpha)}(X)}(\hat{\theta}(X) - \theta)\right]c^t \ge c\{\lambda(\theta)G^{(\alpha)}_\lambda\}^{-1}c^t.
  \end{align}
  Integrating the above over $\theta$, we get
  \begin{align}
  \label{eq:analogous_cramer_rao_inequality3}
  c\int \text{Var}_{\theta^{(\alpha)}}\left[\frac{\tilde{p_\theta}(X)}{p_\theta^{(\alpha)}(X)}(\hat{\theta}(X) - \theta)\right] d\theta~c^t\nonumber\\
    \ge c~ \int [\lambda(\theta)G^{(\alpha)}_\lambda]^{-1} d\theta ~ c^t.
  \end{align}
  But
  \begin{equation}
      \int [\lambda(\theta)G^{(\alpha)}_\lambda]^{-1} d\theta \ge \big\{\mathbb{E}_\lambda [G_\lambda^{(\alpha)}(\theta)]\big\}^{-1}
  \end{equation}
  by \cite{GrovesRothenberg1969Biometrika}. This proves the result.
\end{proof}

The above result reduces to the usual Bayesian Cramer-Rao inequality when $\alpha = 1$ as in \cite{kumar2018information}. When $\lambda$ is the uniform distribution, we obtain the $\alpha$-Cramer-Rao inequality as in \cite{kumar2020cram}. When $\alpha = 1$ and $\lambda$ is the uniform distribution, this yields the usual deterministic Cramer-Rao inequality.

\section{Information Geometry for Hybrid CR inequality}
\label{sec:hybrid}
Hybrid CR inequality is a special case of Bayesian CR inequality where part of the unknown parameters are deterministic and the rest are random. This was first encountered by Rockah in a specific application \cite{rockah1987arrayfar,rockah1987arraynear}. Further properties of hybrid CR inequality were studied, for example, in \cite{narasimhan1995fundamental,noam2009notes, messer2006hybrid}. %The setting in the hybrid case is the following. 

Consider the setting in \ref{sec:bcrlb}. The unknown parameter $\theta$ is now concatenation of two vectors $\theta_1$ and $\theta_2$, that is, $\theta = [\theta_1^T, \theta_2^T]^T$, where $\theta_1$ is an $m$-dimensional vector of deterministic parameters and $\theta_2$ is an $n$-dimensional vector of random parameters. Since $\theta_1$ is deterministic, the prior distribution $\lambda(\theta)$ is independent of $\theta_1$. As a consequence, the entries $J_{i,j}^{\lambda}$ in \eqref{eqn:J_matrix_bayesian} corresponding to any of the components of $\theta_1$ vanish. The hybrid CR inequality takes a form that is same as the Bayesian one except that the $J^{\lambda}$ matrix in \eqref{eqn:Bayesian_alpha_cramerrao} now becomes
\begin{align}
\begin{pmatrix}
  \begin{matrix}
  0
  \end{matrix}
  & 0 \\
  0 &
  \begin{matrix}
  J^{\lambda}(\theta_2),
  \end{matrix}
\end{pmatrix}
\end{align}
where $J^{\lambda}(\theta_2)$ is the $J^{\lambda}$ matrix for the random parameter vector $\theta_2$. In a similar way, one obtains the hybrid $\alpha$-CR inequality from Theorem \ref{thm:Bayesian_alpha_CRLB}.

\section{Summary}
\label{sec:summary}
In this chapter, we discussed information-geometric characterizations of various divergence functions linking them to the classical $\alpha$-CRLB, generalized CRLB, Bayesian CRLB, Bayesian $\alpha$-CRLB, hybrid CRLB, and hybrid $\alpha$-CRLB (see Table~\ref{tbl:summary}). For the Bayesian CRLB, we exploited the definition of KL-divergence when the probability densities are not normalized. This is an improvement over Amari-Nagaoka framework \cite{amari2000methods} on information geometry which only dealt with the notion of deterministic classical CRLB. 

In particular, we formulated an analogous inequality from the generalized Csisz\'ar $f$-divergence. This result leads the usual CR inequality to its escort $F(p)$ by the transformation $p\mapsto F(p)$. Note that this reduction is not coincidental because the Riemannian metric derived from all Csisz\'ar $f$-divergences is the Fisher information metric and the divergence studied here is a Csisz\'ar $f$-divergence, not between $p$ and $q$, but between $F(p)$ and $F(q)$. The generalized version of the CR inequality enables us to find unbiased and efficient estimators for the escort of the underlying model. %This theory when specified to relative $\alpha$-entropy, gives rise to an $\alpha$-CRLB for the $\alpha$-escort of the underlying distribution. 

    %----------------------------------------------------------------------------------
    \begin{table}
    \caption{Lower error bounds and corresponding information-geometric properties}
    \label{tbl:summary}       % Give a unique label
    %
    % Follow this input for your own table layout
    %
    \begin{tabular}{p{3.1cm}P{1.3cm}P{1.6cm}P{3.5cm}}
    \hline\noalign{\smallskip}
    Bound & cf. Section & Divergence & Riemannian metric \\
    \noalign{\smallskip}
    \hline
    \noalign{\smallskip}
    Deterministic CRLB \cite{amari2000methods}  & \ref{subsec:intro_relent}, \ref{sec:desiderata} & ${I}(p,q)$ & $G^{(e)}(\theta)$ \\
    Bayesian CRLB \cite{kumar2018information} & \ref{sec:bcrlb} & $I(\tilde{p}_\theta\|\tilde{p}_{\theta'})$  & $\lambda(\theta)\big[G^{(e)}(\theta) + J^{\lambda}(\theta)\big]$  \\
    Hybrid CRLB & \ref{sec:hybrid} & $I(\tilde{p}_\theta\|\tilde{p}_{\theta'})$  & $\lambda(\theta)\big[G^{(e)}(\theta) + J^{\lambda}(\theta)\big]$ \\
    Barankin bound \cite{kumar2018information} & \ref{sec:bcrlb} & Not applicable  &  $g(\theta)$\\
    Deterministic $\alpha$-CRLB \cite{kumar2020cram} & \ref{subsec:analogous_cr_inequality} & $I_{\alpha}(p,q)$ & $G^{(\alpha)}(\theta)$\\
    General $(f, F)$-CRLB \cite{kumar2020cram} & \ref{subsec:general_framework} & $D_f^{(F)}(p,q)$  & $G^{(f, F)}(\theta)$\\
    Bayesian $\alpha$-CRLB \cite{mishra2020generalized} & \ref{sec:relalphaBayes} & $I_{\alpha}(\tilde{p}_{\theta},\tilde{p}_{\theta'})$
     & $\lambda(\theta)\big[G^{(\alpha)}(\theta) + J^{\lambda}(\theta)\big]$ \\
    Hybrid $\alpha$-CRLB & \ref{sec:hybrid} & $I_{\alpha}(\tilde{p}_{\theta},\tilde{p}_{\theta'})$
     & $\lambda(\theta)\big[G^{(\alpha)}(\theta) + J^{\lambda}(\theta)\big]$ \\
    \noalign{\smallskip}\hline\noalign{\smallskip}
    \end{tabular}
    \end{table}
    %----------------------------------------------------------------------------------
    
Finally, using the general definition of $I_\alpha$-divergence in the Bayesian case, we derived Bayesian $\alpha$-CRLB and hybrid CRLB. These improvements enable usage of information-geometric approaches for biased estimators and noisy situations as in radar and communications problems \cite{mishra2017performance}. 

\section*{Acknowledgements}
The authors are sincerely grateful to the anonymous reviewers whose valuable comments greatly helped in improving the manuscript. K. V. M. acknowledges support from the National Academies of Sciences, Engineering, and Medicine via Army Research Laboratory Harry Diamond Distinguished Postdoctoral Fellowship.
  
  \appendix
\chapter{Other Generalizations of Cram\'er-Rao Inequality}
\label{app:general}
  Here we discuss commonalities of some of the earlier generalizations of CR inequality with the $\alpha$-CR inequality mentioned in Section~\ref{subsec:analogous_cr_inequality}.
  
  \begin{enumerate}
      \item Jan Naudts suggests an alternative generalization of the usual Cram\'er-Rao inequality in the context of Tsallis' thermostatistics \cite[Eq.~(2.5)]{naudts2004estimators}.
  Their inequality is closely analogous to ours. It enables us to find a bound for the variance of an estimator of the underlying model (\textit{with respect to the escort model}) in terms of a generalized Fisher information ($g_{k l}(\theta)$) involving both the underlying ($p_\theta$) and its escort families ($P_\theta$).
  Their Fisher information, when the escort is taken to be $P_\theta = p_\theta^{(\alpha)}$, is given by
  \begin{align*}
  g_{k,l}(\theta) = \sum_x \frac{1}{p_\theta^{(\alpha)}(x)} \partial_k p_\theta(x)\partial_l p_\theta(x).
  \end{align*}
  The same in our case is
  \begin{align*}
  g_{k,l}^{(\alpha)}(\theta) = \sum_x \frac{1}{p_\theta^{(\alpha)}(x)} \partial_k p_\theta^{(\alpha)}(x)\partial_l p_\theta^{(\alpha)}(x).
  \end{align*}
  Also, $\partial_i p_\theta^{(\alpha)}$ and $\partial_i p_\theta$ are related by 
  \begin{align*}
  \partial_i p_\theta^{(\alpha)}(x) = \partial_i\left(\frac{p_\theta(x)^\alpha}{\sum_y p_\theta(y)^\alpha}\right) &= \alpha\left[\frac{{p_{\theta}^{(\alpha)}(x)}}{p_{\theta}(x)}\partial_i p_{\theta}(x) \right.\nonumber\\
  &\hspace{8mm}\left.- p_{\theta}^{(\alpha)}(x) \sum_y \frac{{p_{\theta}^{(\alpha)}(y)}}{p_{\theta}(y)}\partial_i p_{\theta}(y)\right].
  \end{align*}
  Moreover, while theirs bounds the variance of an estimator of the \textit{true distribution} with respect to the escort distribution, ours bounds the variance of an estimator of the \textit{escort distribution itself}. Their result is precisely the following.
  
  \vspace{0.2cm}
\noindent
\textbf{Theorem 2.1 of Jan Naudts \cite{naudts2004estimators}} \textit{Let be given two families of pdfs $\left(p_{\theta}\right)_{\theta \in D}$ and $\left(P_{\theta}\right)_{\theta \in D}$ and corresponding expectations $E_{\theta}$ and $F_{\theta} .$ Let c be an estimator of $\left(p_{\theta}\right)_{\theta \in D},$ with scale function $F$. Assume that the regularity condition
\begin{align*}
F_{\theta} \frac{1}{P_{\theta}(x)} \frac{\partial}{\partial \theta^{k}} p_{\theta}(x)=0,
\end{align*}
holds. Let $g_{k l}(\theta)$ be the information matrix introduced before. Then, for all u and v in $\mathbb{R}^{n}$ is
\begin{align*}
\frac{u^{k} u^{l}\left[F_{\theta} c_{k} c_{l}-\left(F_{\theta} c_{k}\right)\left(F_{\theta} c_{l}\right)\right]}{\left[u^k v^l\frac{\partial^2}{ \partial\theta^l\partial\theta^k} F(\theta)\right]^{2}} \geq \frac{1}{v^{k} v^{l} g_{k l}(\theta)}.
\end{align*}
%The bound is optimal (in the sense that equality holds whenever $u=v$ ) if there exist a normalization function $Z>0$ and a function $G$ such that
%\begin{align}
%\frac{\partial}{\partial \theta^{k}} p_{\theta}(x)=Z(\theta) P_{\theta}(x) \frac{\partial}{\partial \theta^{k}}\left[G(\theta)-\theta^{l} c_{l}(x)\right],
%\end{align}
%holds for all $k$ in $[1, \ldots, m],$ for all $\theta \in D,$ and for $\mu$ -almost all $x$. In that case, c is an estimator $o f\left(P_{\theta}\right)_{\theta \in D}$ with scale function $G$ such that $F_{\theta} c_{k} = {\partial G}/{\partial \theta^{k}}$.
}
\vspace{0.2cm}

\item Furuichi \cite{furuichi2009on} defines a generalized Fisher information based on the $q$-logarithmic function and gives a bound for the variance of an estimator with respect to the escort distribution. Given a random variable $X$ with the probability density function $f(x)$, they define the $q$-score function $s_{q}(x)$ based on the $q$-logarithmic function and $q$-Fisher information $J_{q}(X) = E_{q}\left[s_{q}(X)^{2}\right]$,
where $E_{q}$ stands for expectation with respect to the escort distribution $f^{(q)}$ of $f$ as in \eqref{eqn:escort_distribution}. Observe that
\begin{align}
\label{eqn:furuchi_fisher}
    J_{q}(X) & = E_{q}\left[s_{q}(X)^{2}\right]\nonumber\\
    & = E_{q}\left[f(X)^{2-2q}\left(\frac{d}{dX}\log f(X)\right)^2\right],
\end{align}
whereas our Fisher information in this setup, following \eqref{eqn:g-alpha-expansion}, is
\begin{equation}
\label{eqn:q-fisher_information}
    g^{(q)}(X) = E_{q}\left[ \left(\frac{d}{dX}\log f(X)\right)^2\right] - \left(E_{q}\left[ \frac{d}{dX}\log f(X)\right]\right)^2,
\end{equation}
Interestingly, they also bound the variance of an estimator of the escort model with respect to the escort model itself as in our case. Their main result is the following.
\vspace{0.2cm}

\noindent
\textbf{Theorem 3.3 of Furuichi \cite{furuichi2009on}}: \textit{Given the random variable $X$ with the probability density function $p(x)$, the $q$-expectation value $\mu_{q} = E_{q}[X]$, and the $q$-variance $\sigma_{q}^{2} = E_{q}\left[\left(X-\mu_{q}\right)^{2}\right]$, we have a $q$-Cramér-Rao inequality
\begin{align*}
J_{q}(X) \geq \frac{1}{\sigma_{q}^{2}}\left(\frac{2}{\int p(x)^{q} d x}-1\right) \quad \text { for } q \in[0,1) \cup(1,3).
\end{align*}
Immediately, we have
\begin{align*}
J_{q}(X) \geq \frac{1}{\sigma_{q}^{2}} \quad \text { for } q \in(1,3).
\end{align*}
}
\vspace{0.2cm}

\item Lutwak et al. \cite{lutwak2005cramer} derives a Cram\'er-Rao inequality in connection with extending Stam's inequality for the generalized Gaussian densities. Their inequality finds lower bound for the $p$-th moment of the given density ($\sigma_p[f]$) in terms of a generalized Fisher information. Their Fisher information $\phi_{p, \lambda}[f]$, when specialised to $p=q=2$, is given by
\[
\phi_{2, \lambda}[f] = \left\{E\Big[f(X)^{2\lambda - 2}\Big(\frac{d}{dX}\log f(X)\Big)^2\Big]\right\}^{\frac{1}{2}},
\]
which is closely related to that of Furuichi's \eqref{eqn:furuchi_fisher} upto a change of measure $f\mapsto f^{(\lambda)}$, which, in turn, related to ours \eqref{eqn:q-fisher_information}. Moreover, while they use $I_\alpha$-divergence to derive their moment-entropy inequality, they do not do so while defining their Fisher information and hence obtain a different Cram\'er-Rao inequality. Their result is reproduced as follows.
\vspace{0.2cm}

\noindent
\textbf{Theorem 5 of Lutwak et al. \cite{lutwak2005cramer}}: \textit{Let $p \in[1, \infty], \lambda \in(1 /(1+p), \infty),$ and $f$ be a density. If $p<\infty,$ then $f$ is assumed to be absolutely continuous; if $p=\infty,$ then $f^{\lambda}$ is assumed to have bounded variation. If $\sigma_p[f], \phi_{p, \lambda}[f]<\infty,$ then
\begin{align*}
\sigma_p[f]\phi_{p, \lambda}[f] \geq \sigma_p[G]\phi_{p, \lambda}[G],
\end{align*}
where $G$ is the generalized Gaussian density.
}
\vspace{0.2cm}

\item Bercher \cite{bercher2012generalized} derived a two parameter extension of Fisher information and a generalized Cram\'{e}r-Rao inequality which bounds the $\alpha$ moment of an estimator. Their Fisher information, when specialised to $\alpha = \beta = 2$, reduces to
\begin{align*}
I_{2, q}[f ; \theta] = E_q\left[\frac{f^{(q)}(X ; \theta)}{f(x ; \theta)} \left(\frac{\partial}{\partial \theta} \log f^{(q)}(X ; \theta)\right)^2\right],
\end{align*}
where $E_q$ stands for expectation with respect to the escort distribution $f^{(q)}$. Whereas, following \eqref{eqn:RiemannianOnS-alpha}, our Fisher information in this setup is
\begin{align*}
    g^{(q)}(\theta) = \frac{1}{q^2} E_{q}\left[\left(\frac{\partial}{\partial \theta} \log f^{(q)}(X ; \theta)\right)^2\right].
\end{align*}
Thus our Fisher information differs from his by the factor ${f(x;\theta)}/{q^2 f^{(q)}(x,\theta)}$ inside the expectation. Note that $q$ in their result is analogous to $\alpha$ in our work. The main result of Bercher \cite{bercher2012generalized} is reproduced verbatim as follows.
\vspace{0.2cm}

\noindent
\textbf{Theorem 1 of Bercher \cite{bercher2012generalized}}: \textit{Let $f(x; \theta)$ be a univariate probability density function defined over a subset $X$ of $\mathbb{R},$ and $\theta \in \Theta$ a parameter of the density. Assume that $f(x ; \theta)$ is a jointly measurable function of $x$ and $\theta$, is integrable with respect to $x$, is absolutely continuous with respect to $\theta,$ and that the derivative with respect to $\theta$ is locally integrable. Assume also that
$q>0$ and that $M_{q}[f ; \theta]$ is finite. For any estimator $\hat{\theta}(x)$ of $\theta,$ we have
\begin{align*}
E\left[|\hat{\theta}(x)-\theta|^{\alpha}\right]^{\frac{1}{\alpha}} I_{\beta, q}[f ; \theta]^{\frac{1}{\beta}} \geq\left|1+\frac{\partial}{\partial \theta} E_{q}[\hat{\theta}(x)-\theta]\right|,
\end{align*}
with $\alpha$ and $\beta$ H\"{o}lder conjugates of each other, i.e., $\alpha^{-1}+\beta^{-1}=1, \alpha \geq 1,$ and where the quantity
\begin{align*}
I_{\beta, q}[f ; \theta] = E\left[\left|\frac{f(x ; \theta)^{q-1}}{M_{q}[f ; \theta]} \frac{\partial}{\partial \theta} \ln \left(\frac{f(x ; \theta)^{q}}{M_{q}[f ; \theta]}\right)\right|^{\beta}\right],
\end{align*}
where $M_{q}[f ; \theta] := \int f(x ; \theta)^{q}~dx$, is the generalized Fisher information of order $(\beta, q)$ on the parameter $\theta .$ %The equality case is obtained if
%\begin{align}
%\frac{q}{M_{q}[f ; \theta]^{\frac{1}{q}}} \frac{\partial}{\partial \theta} \ln _{q *}\left(\frac{f(x ; \theta)}{M_{q}[f ; \theta]}\right)=c(\theta) \operatorname{sign}(\hat{\theta}(x)-\theta)|\hat{\theta}(x)-\theta|^{\alpha-1},
%\end{align}
%with $c(\theta)>0$.
}
\end{enumerate}
\Backmatter
%
%%%%%%%%% REFERENCES %%%%%%%%%
% 
%        ====  Option 1 ====
%        Plain LaTeX entries
%        ===================
% 
% \include{references}
%
%
%         ==== Option 2 =====
%         bibtex (+natbib)
%         ===================
%
% Numbered style.
% Set \citestyle{elsarticle-num} in document preamble.
\bibliographystyle{elsarticle-num}
\bibliography{book}  
%
% Harvard (name-date) style.
% Set \citestyle{elsarticle-harv} in document preamble.
% \bibliographystyle{elsarticle-harv}
%
% Your own bibstyle.
% \bibliographystyle{<your-bib-style>}
%
% Bibdata (your-bibtex-file.bib)
% \bibliography{<your-bib-file>}
%
%
%
%           ==== Option 3 =====
%           biblatex
%           ===================
% Remove 'natbib' package from document preamble and
% load 'biblatex' package instead.
%
% \printbibliography
% 
\end{document}